\documentclass[12pt,english]{article}
\usepackage[T1]{fontenc}
\include{KAM-type-bib}
\usepackage[latin9]{inputenc}
\usepackage{geometry}
\usepackage{amsthm}
\usepackage{amsmath}
\usepackage{amssymb}
\usepackage{graphicx}
\usepackage{setspace}
\usepackage{esint}
\theoremstyle{plain}
\newtheorem{thm}{\protect\theoremname}[section]
  \theoremstyle{plain}
  \newtheorem{lem}[thm]{\protect\lemmaname}
  \theoremstyle{remark}
  \newtheorem*{rem*}{\protect\remarkname}
  \theoremstyle{definition}
  \newtheorem{defn}[thm]{\protect\definitionname}
  \theoremstyle{plain}
  \newtheorem{prop}[thm]{\protect\propositionname}
  \theoremstyle{plain}
  \newtheorem{cor}[thm]{\protect\corollaryname}
  \theoremstyle{remark}
  
  \theoremstyle{definition}
  \newtheorem*{example*}{\protect\examplename}

\usepackage{changepage}
\usepackage{caption}
\makeatother

\usepackage{babel}
  \providecommand{\claimname}{Claim}
  \providecommand{\corollaryname}{Corollary}
  \providecommand{\definitionname}{Definition}
  \providecommand{\examplename}{Example}
  \providecommand{\lemmaname}{Lemma}
  \providecommand{\propositionname}{Proposition}
  \providecommand{\remarkname}{Remark}
\providecommand{\theoremname}{Theorem}
\usepackage{color}

\begin{document}

\date{\today}

\begin{title}
{On near integrability of some impact systems}
\end{title}
\author{M. Pnueli$^1$, V. Rom-Kedar$^{1,2}$ \\
\normalsize $^1$ Department of Computer Science and Applied Mathematics, \\ \normalsize The Weizmann Institute of Science, Rehovot, Israel \\
\normalsize $^2$ The Estrin Family Chair of Computer Science and Applied Mathematics.
}
\date{\today}
\maketitle

\begin{abstract}
A class of Hamiltonian impact systems exhibiting smooth near integrable behavior is presented.  The underlying unperturbed model investigated is an integrable, separable, 2 degrees of freedom mechanical impact
system with effectively bounded energy level sets and a single straight
wall which preserves the separable structure. Singularities in the system appear either as trajectories with tangent impacts or as singularities in the underlying Hamiltonian
structure (e.g.  separatrices). It is shown that away from these singularities,
a small perturbation from the integrable structure results in smooth
near integrable behavior. Such a perturbation may occur from a small
deformation or tilt of the wall which breaks the separability upon impact, the addition
of a small regular perturbation to the system, or the combination
of both. In some simple cases explicit formulae to the leading order term in the near integrable return map are derived.  Near integrability is also shown to persist when the hard billiard boundary is replaced by a singular, smooth, steep potential, thus extending the near-integrability results beyond the scope of regular perturbations.
These systems  constitute an additional class of examples of near integrable impact systems, beyond the traditional one dimensional oscillating billiards, nearly elliptic billiards,  and the near-integrable behavior near the boundary of convex smooth billiards with or without magnetic field.\end{abstract}

\maketitle


\section{Background}

The global phase-space structure of smooth nonlinear \(n\) d.o.f. Hamiltonian systems with \(n\geq2\) is usually unknown. While numerical simulations for such systems are readily available, they are usually difficult to interpret due to our limited perception of high dimensional spaces. Moreover, the abundance of various phase space structures in such systems (tori, cantori, homoclinic tangencies, lower dimensional whiskered tori etc.), shadowed by  chaotic solutions, complicates the dynamics and its averaged and asymptotic expressions. One therefore seeks to study special classes of systems which are amenable to analysis in some limit and inspire the definition of particular observables and projections that detect the closeness of the given system to its limiting behavior. Traditional examples are near-integrable systems and slow-fast systems \cite{Arnold2013mathematicalmethods,Arnold2007CelestialMechanics,haller2012chaos}.
More recently, analytical tools for studying smooth near-billiard and near-impact systems have been developed \cite{kozlov1991billiards,rom2012billiards,rapoport2007approximating,lerman2012saddle,RK2014smooth}. In these works, the limit system is  a Hamiltonian Impact System (HIS) which describes the dynamics
of a particle moving under the influence of a potential
inside a domain and reflecting elastically from its boundary. Billiards correspond to the simplest HIS with inertial motion (trivial potential) in the domain interior. By this approach, to better understand systems with very steep potentials at the domain's boundary, one studies the limit system in which the steep part is replaced by impacts. Once the dynamics under the HIS are known, one  establishes which of its features persist \cite{rapoport2007approximating,RK2014smooth} and how those which do not persist bifurcate \cite{rom2012billiards,rapoport2008stability}.

The study of HIS combines the features of Hamiltonian dynamics and those of piecewise
smooth  dynamical systems \cite{bernardo2008piecewise,kunze1997application,makarenkov2012dynamics}, which are specific examples of hybrid systems (e.g. \cite{kazakov2013dynamical,granados2014scattering}).  Utilizing the Hamiltonian structure, one hopes to gain information on global scales. Yet, impacts destroy the smoothness  \cite{chernov2006chaotic,kozlov1991billiards} and possibly the integrability of the underlying Hamiltonian flow \cite{lerman2012saddle}.  Finding integrable HIS  and studying their behavior under perturbations (of the boundary and of the potential) expands the families of non-linear systems which we can analyse and, by utilizing the smooth impact framework, allows to establish near-integrablity results even though the perturbation terms in this case are formally large in the \(C^{r }\) topology. Here we provide such a class of prototype impact systems which are near-integrable and are amenable to analysis.
Previous near-integrability results for HIS have utilized the local dynamics near periodic orbits \cite{dullin1998linear,berglund1996billiards,berglund2000classical,RK2014smooth}, near a smooth convex boundary \cite{zharnitsky2000invariant,berglund1996integrability,berglund2000classical} and near saddle-center homoclinic connection of a quadratic potential with impacts \cite{lerman2012saddle}. Another approach utilized the generalized adiabatic theory in 1.5 d.o.f. systems, where the Hamiltonian dynamics are in one dimension and the boundary is slowly oscillating \cite{neishtadt2008jump,gorelyshev2006adiabatic,neishtadt2012destruction,artemyev2015violation} . Similar approaches were employed in the study of magnetic billiards \cite{robnik1985classical,berglund1996integrability,berglund1996billiards,berglund2000classical}.

Here we address the subject of  global structure and  stability of orbits on large portions of the phase space by identifying regimes in which  standard smooth near-integrable results apply (in particular, persistence of KAM tori and the emergence of resonances). These objects arise  even though, formally, we are far from the classical setup of smooth small perturbations to smooth integrable systems.
To this aim, we focus on
2 degrees of freedom mechanical impact systems, where the underlying
Hamiltonian is of the form $H=\frac{p_{1}^{2}}{2}+\frac{p_{2}^{2}}{2}+V(q_{1},q_{2})$ and \(V(\cdot)\)
is a separable,
smooth ($C^{r+1}$) potential with effectively bounded level sets. The impact in the system is realized as a single straight vertical wall, where the seperability assumption is with respect to the vertical wall coordinate system, so that this wall does not destroy
integrability. A perturbation from the integrable structure is then realized by either the addition of a small, $\mathcal{O}(\epsilon_r)$, $C^{r+1}$ regular coupling perturbation to the potential or a small $\mathcal{O}(\epsilon_w)$, $C^{r+1}$ deformation of the wall. The main result here is that under some specific conditions, in a large (\(\mathcal{O}(1)\) measure) phase space region, smooth near-integrable dynamics are realized for sufficiently  small \(\epsilon_{r}\) and  \(\epsilon_{w}\).  Moreover, using \cite{RK2014smooth}, it is shown that these results may be extended for the smooth system in which the hard wall is replaced by a soft steep potential, provided the potential is sufficiently steep (notably, the steeper the potential is the larger the perturbation is in the \(C^{r+1}\) topology).

The paper is organized as follows. In section \ref{sec:The-Integrable-System},
the underlying integrable structure of the systems investigated is
presented, and an integrable Poincar\'{e} return map is constructed. Conditions for smoothness of the return map are shown,
as well as the conditions for twist and for resonance. In Section
\ref{sec:Near-integrability-results}, it is shown how, following
the conclusions of section \ref{sec:The-Integrable-System}, one can achieve near integrability results when adding a small regular
perturbation to the system or when considering a small deformation of the
wall from the vertical, perpendicular position. Furthermore, it is shown that for the case of straight walls explicit formulae for the leading order terms of this return map in the wall inclination and the smooth perturbation term may be calculated as Melnikov-type integrals.  Near integrability is also extended to the corresponding soft impact systems. An example to the main results is given in section
\ref{sec:Examples}, where, additionally, the global perturbed phase-space structure is presented in an impact-energy-momentum diagram. We summarize our results in section \ref{sec:Conclusions}.

\section{Setup and integrability results} \label{sec:The-Integrable-System}


Consider a 2 degrees of freedom mechanical impact
system of the form:
\begin{equation}
H=H(\cdot;\epsilon_{r},\epsilon_{w},q^{w},b)=H_{int}(q_{1},p_{1},q_{2},p_{2})+\epsilon_{r}V_{r}(q_{1},q_{2})+b\cdot{}V_{b}(q_{}-q_{}^{w};\epsilon_{w})
\label{eq:hgeneral}
\end{equation}
where the underlying integrable structure is separable (see below), the potential $V_{r}(q_{1},q_{2})$
 corresponds to a regular smooth ($ C^{r+1}$ with $r>3$) coupling term and the singular billiard potential \( V_{b}(q_{}-q_{}^{w};\epsilon_{w})\) represents the singular impact term. Hereafter, for \(\epsilon_{w}=0\), the impact corresponds to a single vertical wall  passing through the origin (namely, with no loss of generality, the \(q_{2}\) axis is set along the wall and the origin is set at some point on the wall, otherwise shift the $q_2$ coordinate by a constant value). A non-zero \(\epsilon_{w}\) corresponds to small perturbations, in $ C^{r+1}$, from the vertical geometry, so \(q^{w}=(q_{1}^{w}=\epsilon_{w}Q^{w}(q_{2}^w;\epsilon_{w}),q_{2}^w)\), and \(Q^{w}\) is a \(C^{r+1}\) function satisfying \(Q^{w}(0;0)=0\). Motion occurs to the right of the wall;  the wall is realized in
the system as a singular energy barrier:
\begin{equation}
V_b=\begin{cases}
0, & (q_1,q_2):q_1>\epsilon_{w}Q^{w}(q_{2};\epsilon_{w})\\
1, & (q_1,q_2):q_1\leq\epsilon_{w}Q^{w}(q_{2};\epsilon_{w})
\end{cases}\label{eq:billiardpot}
\end{equation}
and $b$ is either a fixed large number or zero (when positive it is taken such that for all energies of interest the wall is impassable, whereas $b=0$ refers to the smooth Hamiltonian system without the impact).
  
 The integrable structure of \(H_{int}\) of (\ref{eq:hgeneral}) is of the form:\begin{equation}
H_{int}=\frac{||p||^{2}}{2}+V_{int}(q_{1},q_{2})=\frac{p_{1}^{2}}{2}+V_{1}(q_{1})+\frac{p_{2}^{2}}{2}+V_{2}(q_{2})
=H_{1}(q_{1},p_{1})+H_{2}(q_{2},p_{2}),
\label{eq:hint}
\end{equation}
where the potential $V_{int}=V_{1}(q_{1})+V_{2}(q_{2})$ is Separable,   $C^{r+1}$ ($r>3 $) Smooth, Simple (has finite, discrete number of simple
extremum points), Bounded from below and $V_i$ go to infinity
as $|q_{i}|\rightarrow\infty$, so $H_{int}$ has only bounded level
sets. Therefore the perturbation terms \(V_{r}(\cdot),Q^w(\cdot)\) are bounded on the energy surfaces (see appendix), where the bound depends on \(H\). Hereafter we assume that \(H=\mathcal{O}(1)\) - the asymptotic behavior at large \(H \) may require additional analysis.  
\begin{defn} Integrable Hamiltonians of the form (\ref{eq:hint}) satisfying the above conditions will  hereafter be called \emph{Hamiltonians of the S3B (Separable, Smooth, Simple, Bounded level sets) class}.\end{defn}



Next we define the phase space region for which the results apply. 
We first describe the smooth integrable structure.
Denote, in each sub phase space $i=1,2$,
the $n_{ic}$ center fixed points by $(q_{ic,1...n_{ic}},0)$ and
the $n_{is}$ saddle fixed points by $(q_{is,1...n_{is}},0)$. Let \(\mathbb{I}_{H_{i}}(H)\) denote the set of allowed \(H_{i}\) values for a given \(H\) (here, the interval \(H_i\in[\min V_{i},H-\min V_{\bar i}]\), where \(\bar i\) denotes the complement phase space to \(i\)), and let \(\mathcal{N}_{\delta}(H_i^{*})\) denote the \(\delta-\) open interval of \(H_i\) values around \(H_i^*\).
 The closed set of regular integrable \(H_{i}\) values
on a given energy level \(H\) 

\begin{equation}
\mathbb{H}_{i}^{R,\delta}(H) = \mathbb{I}_{H_{i}}(H)\backslash[\bigcup_{j=1}^ {n_{ is}}
\mathcal{N}_{\delta}(H_{i}(q_{is,j},0))\cup\bigcup_{j=1}^{n_{\bar is}}
\mathcal{N}_{\delta}(H-H_{\bar i}(q_{\bar i s,j},0))], \ \delta>0,i=1,2\label{eq:regularHiset}
\end{equation}
corresponds to \(H_{i}\) values for which the Liouville leaves are bounded away from singularities, namely the energy of the level sets in both the $(q_i,p_i)$ plane and in the $(q_{\bar{i}}, p_{\bar{i}})$ plane  are at least $\delta-$ away from the energies of the  planar singular level sets of 
the saddle points $q_{is},q_{\bar i s}$ respectively (hereafter, normally elliptic lower dimensional tori are included in the regular set). Clearly, the measure of these $H_i$ intervals is of \(\mathcal{O}(1)\) when \(\delta\rightarrow0\) :
\begin{equation}
|\mathbb{H}_{i}^{R,\delta}(H)|\geq H-\min V_{i}-\min V_{\bar i}-2\delta(n_{ is}+n_{ \bar is}).\label{eq:lengthofHi}
\end{equation}

Using the local action-angle variables for the smooth unperturbed integrable system (\( \epsilon_{r}=0, b=0\)), for all $H_{2}\in \mathbb{H}_{2}^{R,\delta}(H)$, $H_{int}$
can be written as $H_{int}(J,I)=H_{1}(J)+H_{2}(I)$, and the dynamics
on the corresponding leaves of the level sets are described by: \begin{multline}
\begin{cases}
\dot{\varphi}=\frac{\partial H_{int}}{\partial J}=\omega_{1}(J), & \dot{J}=-\frac{\partial H_{int}}{\partial\varphi}=0\\
\dot{\theta}=\frac{\partial H_{int}}{\partial I}=\omega_{2}(I), & \dot{I}=-\frac{\partial H_{int}}{\partial\theta}=0
\end{cases}\Rightarrow\begin{cases}
\varphi(t)=\varphi_{0}+\omega_{1}(J_{0})\cdot t & J(t)=J_{0}\\
\theta(t)=\theta_{0}+\omega_{2}(I_{0})\cdot t & I(t)=I_{0}
\end{cases}
\end{multline}
where $(J_{0},\varphi_{0})=S_{1}(q_{1}(0),p_{1}(0)),(I_{0},\theta_{0})=S_{2}(q_{2}(0),p_{2}(0))$ and \(S_i\) denote  local transformations to action-angle coordinates on each leaf.
A branch of the Liouville folliation corresponds to a family of regular leaves (here each leaf is a torus, each branch a one parameter family of tori). On each branch, away from the branch boundaries, the transformation to action-angle coordinates is smooth  and well defined. For each \(H\), the set $ \mathbb{H}_{2}^{R,\delta}(H)$ is composed of a finite number of closed intervals, each corresponding to a finite number of branches.  Since \(H_{i}\) are mechanical Hamiltonians,  $\omega_{i}(\cdot)>0,i=1,2$, and thus $H_{i}^{-1}(\cdot)$ are  uniquely defined on each branch of the Liouville folliation \cite{lerman1998integrable,fomenko2004integrable,Arnold2007CelestialMechanics}. For all level sets in  \(\mathbb{H}_{2}^{R,\delta}(H),\) by the S3B assumption (no parabolic points), there exists a \(K>0\) such that for all  \(I\in H_2^{-1}(\mathbb{H}_{2}^{R,\delta}(H))\)    the frequencies  \(\omega_{2}(I)=H_{2}'(I)\) are bounded from below:
\begin{equation}
\omega_2\geq{}\left|\frac{K}{\ln(\delta)}\right|\label{eq:omegalowerbnd}
\end{equation}
The notation \(H_2^{-1}(\mathbb{H}_{2}^{R,\delta}(H))\) refers to the multi-valued set defined on all (finite number of) relevant branches, and  $H_2^{-1}$ is well defined even when $\delta\rightarrow0$ (though it may be discontinuous at separatrices). It follows that the measure of the excluded set of action values $I$, similarly to the corresponding set of excluded energies in (\ref{eq:regularHiset}) goes to $0$ as $\delta{}\rightarrow{}0$ - see Theorem \ref{thm:(Smoothness)}.

\subsection{Integrable impact return map}

When the wall is vertical ($b\neq{}0, \epsilon_{w}=0$), namely, it respects the separability symmetry of the underlying integrable Hamiltonian flow, one immediately concludes, by the rule of elastic
reflection and the symmetry of the kinetic energy term, integrability:

\begin{lem}
\label{lem:perp=integrable}(Integrability) When $\epsilon_{r}=\epsilon_w=0$ the dynamics of the impact system are integrable.
\end{lem}
The vertical wall produces the additional singular level sets that correspond to tangent trajectories with \(H_1\) energy \(V_1(0)\).
Generically, such level sets do not coincide with the  singular level sets of \(H_{1}\), namely:
\begin{defn}\label{def:regular-wall-position} The vertical wall position is  regular  if $V_{1}'(q_{1}^{w}=0)\ne 0$ and $V_{1}(q_{1}^{w})\ne V_{1}(q_{1s,j}), j=1,...,n_{1s}$.
\end{defn}

 Since \(H_{1}\) is of the mechanical form, it follows that for any given \(H_{1}\) there exists at most a single Liouville leaf in the \((q_{1},p_1)\) plane which intersects the wall at \(q_{1}=q_1^w=0\), called hereafter the \textit{intersecting leaf}.  In particular, for a family of intersecting leaves, the value of   \(|p_{1}^{w}(H_{1})|= \sqrt{H_{1}-V_1(0)}\) is uniquely defined and is monotone in \(H_{1}\) for all \(H_{1}>H_{1tan}=V_1(0)\) as is the dependence on \(J=H_{1}^{-1}(\cdot)\)   on such intersecting leaves. For fixed energy \(H=H_{1}+H_2\), a branch can be either intersecting (meaning that all leaves of this branch intersect the wall), non-intersecting, or tangent. Namely, the location of the perpendicular wall determines uniquely   \textit{the tangent branch;} for regular wall position, for each energy value $H=h\geq{}V_1(0)+\min{V_2}$, there exists a unique leaf within the level set   \(H_1=H_{1tan}, H_{2}=h-H_{1tan    }\)  ,  at which a tangency occurs. For a fixed energy \(H\), we call the branch corresponding to this leaf \textit{the tangent branch}, and on this branch     \(J\approx J_{tan}=H_{1}^{-1}(V_{1}(0))\) uniquely defines the leaves. In conclusion, we establish:
\begin{lem}
 In the unperturbed vertical wall case (\(\epsilon_r=\epsilon_{w}=0\)) with regular wall position, for the flow restricted to the tangent and intersecting branches, tangency
occurs at \(J_{tan}=H_{1}^{-1}(V_{1}(0))\) whereas impacts occur iff  $J \ge J_{tan}$.
\end{lem}
Hereafter, unless specified otherwise, we consider the dynamics only on the tangent and intersecting branches (for the integrable dynamics all other branches are unaffected by the impact). 

Next, a return map of the integrable impact system is constructed and it is proven that it is  $C^{r}$ smooth and satisfies the twist condition for most initial conditions. 
Since the motion occurs to the right of the vertical wall and the impact occurs whenever $q_{1}=0,$  choosing the cross-section  $\Sigma=\{(q_{1},p_{1}):p_{1}=0,\dot{p_{1}}<0\}$
 ensures that in each iteration of the return map at most a single
collision with the wall occurs. The return map to \(\Sigma \), 
 $ $ for the system
without the impact is simply:
\begin{equation}
\begin{cases}
I'&=I(T_{1}(J))=I\\
\theta'&=\theta(T_{1}(J))=\theta+\omega_{2}(I)\cdot T_{1}(J(H,I))=\theta+\frac{T_{1}(J(H,I))}{T_{2}(I)}\cdot2\pi
\end{cases}\label{eq:int_returnmap}
\end{equation}
where  $J=J(H,I)=H_{1}^{-1}(H-H_{2}(I))$,  \(T_{1}(J)=\frac{2\pi}{\omega_{1}(J)}\) and  $T_{2}(I)=\frac{2\pi}{\omega_{2}(I)}$ are  well defined for \(I\in H_2^{-1}(\mathbb{H}_{2}^{R,\delta}(H))\) for small $\delta>0$.
Similarly, the corresponding  return map $\mathcal{F}_0:(I,\theta)\rightarrow(I',\theta')$
of the integrable impact system is defined for \(I\in H_2^{-1}(\mathbb{H}_{2}^{R,\delta}(H))\) by:

\begin{equation}
\begin{cases}
I'=I\\
\theta'=\theta+\omega_{2}(I)\cdot(T_{1}(J)-\Delta t_{travel}(J))=\theta+2\pi\,\frac{\tilde{T}_{1}(J)}{T_{2}(I)}\equiv\theta+\Theta(I,J(H,I))
\end{cases}\label{eq:int_impact_returnmap}
\end{equation}
 with
\begin{equation}
\Delta t_{travel}(J)=\begin{cases}
2\int_{q_{1min}(J)}^{0}\frac{dq_{1}}{\sqrt{2(H_{1}(J)-V_{1}(q_{1}))}} & \mbox{impact }(J>J_{tan})\\
0 & \mbox{no impact, tangency }(J\leq J_{tan})
\end{cases}
\end{equation}
where $q_{1min}(J)$ is the minimal $q_1$ value on the chosen intersecting leaf that satisfies $V_1(q_{1min})=H_1(J)$, i.e. the leftmost point of the trajectory (outside the billiard). Namely, $\Delta t_{travel}$ is the time of travel outside the billiard
which is lost due to the impact (see Figure \ref{fig:deltattravel}) and $\tilde{T}_{1}(J)=T_{1}(J)-\Delta t_{travel}(J)$
is the new period time in $J$. Generically, we expect that the level set of  \(H_{2}=H-V_{1}(0)\) is regular:
 \begin{defn}\label{def:regular-H-wall} \(H\) is  \(\delta-\)regular with respect to the wall position if $V_{1}(0)\in \mathbb{H}_{1}^{R,\delta}(H)$.
\end{defn}
For regular wall position, for sufficiently small \(\delta\), there are at most a finite
number of \(\delta-\)
intervals for which \(H\) is not regular, corresponding to the energy surfaces at which \(H-V_{1}(0)\) are close to  \(V_2(q_{2s,j})\) for some \(j\).
For $\delta$-regular \(H\) values, denote by $I_{tan}(H)=H_{2}^{-1}(H-V_{1}(0))$ and notice that impacting
trajectories, corresponding to $J>J_{tan}$, yield $I<I_{tan}(H)$. We now establish:

\begin{figure}
\begin{centering}
\includegraphics[scale=0.35]{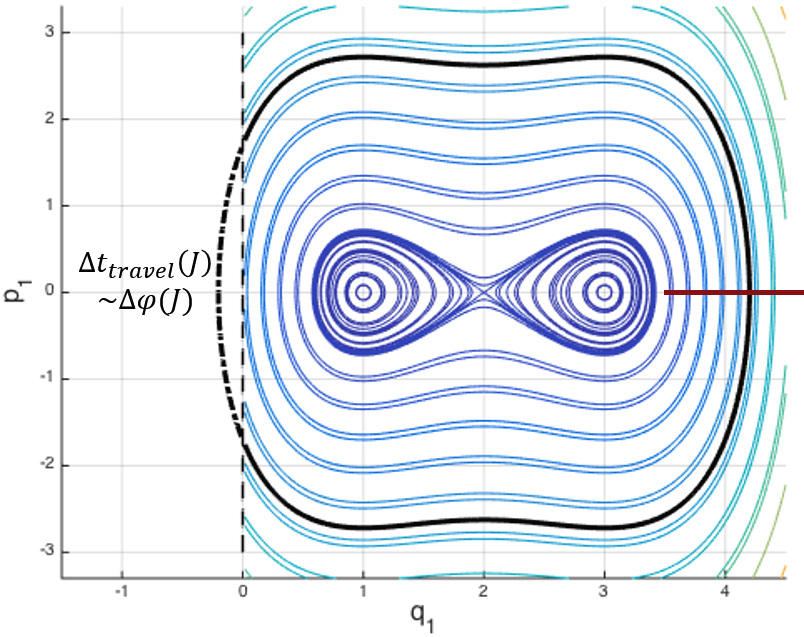}
\par\end{centering}
\protect\caption{\label{fig:deltattravel} Isoenergetic curves in $(q_1,p_1)$, with an impacting trajectory emphasized in black. The wall is at $q_{1}^{w}=0$ (dashed line),  and the cross-section $\Sigma$ corresponding to the relevant intersecting branch is drawn in red. The impact causes a jump in the angle $\Delta\varphi$, which is proportional to the time of travel $\Delta{}t_{travel}$ outside the billiard, i.e. the travel time from $p_1$ to $-p_1$ at the impact point had there been no wall. The Hamiltonian $H_{1}$ in this figure corresponds
to the undamped Duffing oscillator - see section \ref{sec:Examples}.}
\end{figure}

\begin{thm}
\label{thm:(Smoothness)}(Smoothness of (\ref{eq:int_impact_returnmap}))  Consider a Hamiltonian $H$ of the form Eq.
(\ref{eq:hgeneral}) with an S3B  integrable structure \(H_{int}\) and a regular wall position (Def. \ref{def:regular-wall-position}), with $\epsilon_{r}=\epsilon_{w}=0$.
Fix \(\delta>0, \rho>0\), and consider a  \(\delta-\)regular energy level  $H<b$. Then for \(I\) in  \(H_2^{-1}(\mathbb{H}_{2}^{R,\delta}(H))\), excluding  a \(\rho-\) interval centered at \(I_{tan}\)  (so \(I\in H_2^{-1}(\mathbb{H}_{2}^{R,\delta}(H))\backslash\mathcal{N}_{\rho}(I_{tan}(H))\)), the return
map $\mathcal{F}_0:(I,\theta)\rightarrow(I',\theta')$ is   symplectic and $C^{r}$ smooth,  i.e. $\exists M_{r}(\rho,\delta)<\infty$
such that $||\Theta(I,J(H,I))||_{C^{r}}<M_{r}(\rho,\delta)$. Moreover,  the regular set    \( H_2^{-1}(\mathbb{H}_{2}^{R,\delta}(H))\backslash \mathcal{N}_{\rho}(I_{tan}(H))\)
on which 
the return map (\ref{eq:int_impact_returnmap}) is  \(C^{r }\) smooth is of \(\mathcal{O}(1)\) in \(\delta,\rho\).\end{thm}
\begin{proof}
This is a result of the property of smooth dependence on initial conditions
in ODEs and the assumed structure of the flow. For impact away from
tangency, the cross-sections $\{q_{1}=0,p_{1}<0\}$, $\{q_{1}=0,p_{1}>0\}$
are transverse to the flow. The travel time $\Delta t_{travel}(J)$
corresponds to the travel time between the former and the latter transverse
cross-sections. Since  neighborhoods of separatrices are excluded,
this travel time is finite and depends smoothly on initial
conditions. 
Finally, it follows from (\ref{eq:lengthofHi},\ref{eq:omegalowerbnd}) that for small $\delta$, the measure of the set of excluded action values satisfies:
\begin{equation}
|H_2^{-1}(\bigcup_{j=1}^ {n_{2s}}\mathcal{N}_{\delta}(H_{2}(q_{2s,j},0))\cup\bigcup_{j=1}^{n_{1s}}
\mathcal{N}_{\delta}(H-H_{1}(q_{1s,j},0)))\cup \mathcal{N}_{\rho}(I_{tan}(H))|\leq{}\mathcal{O}(\delta{}|\ln(\delta)|,\rho)\label{eq:measIexcluded}
\end{equation}
as for each neighborhood of $q_{1s,j}, j=1,...,n_{1s}$ an $\mathcal{O}(\delta)$ neighborhood of $I$ values is excluded, whereas for each neighborhood of $q_{2s,j}, j=1,...,n_{2s}$ the excluded $I$ values consist of an $\mathcal{O}(\delta|\ln(\delta)|)$ neighborhood (Eq. \ref{eq:omegalowerbnd}).
Hence, as $\delta,\rho\rightarrow{}0$, for \(H=\mathcal{O}(1)  \) as considered here, the set of regular action values \(I\in  H_2^{-1}(\mathbb{H}_{2}^{R,\delta}(H))\backslash \mathcal{N}_{\rho}(I_{tan}(H))\)  is of \(\mathcal{O}(1)\)  as claimed.
\end{proof}

Near tangency the map $\mathcal{F}_0$ is symplectic and  $C^{0}$ - continuous but
not smooth.

\begin{rem*}
For impacting trajectories, the new period time in $J$, $\tilde{T}_{1}(J)=2\int_{0}^{q_{1max}(J)}\frac{dq_{1}}{\sqrt{2(H_{1}(J)-V_{1}(q_{1}))}}$
can also be calculated for level sets near a separatrix when the
saddle point is outside the billiard (so the travel time on the
trajectory inside the billiard is finite), or, similarly, for potentials
with unbounded level sets where the billiard serves to effectively
bound the energy level set (i.e. level sets are unbounded only on
the outer side of the wall). The theorem can therefore be extended
to such cases following suitable alterations to the initial assumptions.
\end{rem*}

The return map (\ref{eq:int_impact_returnmap}) satisfies the twist condition away from the non-twist set: 
\begin{defn}
The Non-Twist set  for a given energy level  $H $ is:
\begin{equation}
\mathbf{I}_{NT}(H)=\{I\in H_2^{-1}(\mathbb{H}_{2}^{R,\delta}(H))\mid H_{1}(J(H,I))+H_{2}(I)=H,\ \frac{d}{dI}(\frac{\tilde{T}_{1}(J)}{T_{2}(I)}\cdot2\pi)=0\}
\end{equation}
\end{defn}

\begin{thm}
\label{thm:(Twist)}(The Regular-Twist set) Consider a Hamiltonian $H$ of the form Eq.
(\ref{eq:hgeneral}) with an S3B  integrable structure \(H_{int}\) and a regular wall position, with $\epsilon_{r}=\epsilon_{w}=0$. Fix \(\delta>0, \rho>0\), and consider a  \(\delta-\)regular energy level  $H<b$. Then, for sufficiently small \(\delta,\rho\), the regular-twist set    \( H_2^{-1}(\mathbb{H}_{2}^{R,\delta}(H))\backslash \mathcal{N}_{\rho}(I_{tan}(H)\cup \mathbf{I}_{NT}(H))\)
on which 
the return map (\ref{eq:int_impact_returnmap}) is a \(C^{r }\) twist map is of \(\mathcal{O}(1)\) in \(\delta,\rho\).\end{thm}
\begin{proof}
Generically, the set  \(\mathbf{I}_{NT}(H)\) is a discrete, finite set, hence excluding its \(\rho-\)intervals leaves, for sufficiently small \(\rho\), a set of \(I \) values of measure of \(\mathcal{O}(1)\). \end{proof}

The set $\mathbf{I}_{NT}(H)\cap\{I>I_{tan}(H)\}$ corresponds to
non-impacting tori 
which are non-twist due to the underlying system, whereas
the set $\mathbf{I}_{NT}(H)\cap\{I<I_{tan}(H)\}$ corresponds\ to tori that lose their twist
 due to the impact.

Using the implicit relation $J=J(H,I)$, after some algebra, the twist condition becomes:\begin{equation}\label{eq:twist_calc}
\frac{d\Theta}{dI}=2\pi\cdot\frac{-T_{1}(J)\cdot\tilde{T}_{1}'(J)-\tilde{T}_{1}(J)\cdot T_{2}'(I)}{T_{2}^{2}(I)}\neq 0
\end{equation}
Since $T_{1}(J),\tilde{T}_{1}(J),T_{2}(I)$ are always non-negative,
a necessary condition to have a non-twist torus is that $\tilde{T}_{1}'(J)$
and $T_{2}'(I)$ have opposite signs (see also section \ref{sec:Examples}).
\begin{prop}
\label{prop:twist-periods}The non-twist set may only occur in regions where the modified periods in each d.o.f. have opposite monotonicity property:  $I\in\mathbf{I}_{NT}(H) \Rightarrow\tilde{T}'_{1}(J)\cdot T_{2}'(I)\leq0$ 
(where $ J=J(H,I)$).
\end{prop}

Finally, notice that the rotation number for the twist map (\ref{eq:int_impact_returnmap})
changes at   \(I=I_{tan}(H)\) from its impacting value   $\frac{\tilde{T}_{1}(J)}{T_{2}(I)}$ to its non-impacting value    $\frac{{T}_{1}(J)}{T_{2}(I)}$, namely  the resonance surfaces change non-smoothly at \(I=I_{tan}\).

\section{Near integrability results} \label{sec:Near-integrability-results}

In the smooth case without the impact, when $\epsilon_{r}\neq0$ is small, the usual near-integrable dynamics emerge, including the existence of KAM tori, resonances near the rational values of \(   \frac{{T}_{1}(J)}{T_{2}(I)}\) and various types of homoclinic chaos near the singular level sets \cite{Arnold2007CelestialMechanics,meyer2008introduction,meiss2007differential,haller2012chaos}. 

Utilizing the construction of the return map (\ref{eq:int_impact_returnmap}) which is an integrable,  \(C^{r}\)\  smooth, symplectic twist map for the vertical wall case, we now show that under  small perturbations \(\epsilon=(\epsilon_{r},\epsilon_w)\neq0,\) the perturbed return map, \(  \mathcal{F}_{\epsilon}\), away from the singularities, is a  \(C^{r}\)-symplectic map (the near singularities behavior will be studied elsewhere). Furthermore, we also establish that this map  is \(C^{r}\)-close to the integrable one, hence KAM theory  applies and invariant near-integrable regions in phase space can be identified. More precisely: 

\begin{thm}\label{thm:near-int}
Consider a Hamiltonian $H$ of the form Eq.
(\ref{eq:hgeneral}) with an S3B  integrable structure \(H_{int}\) and a regular wall position.
Fix \(\delta>0, \rho>0\), let ${\epsilon}=(\epsilon_{r},\epsilon_w)$ and \(\varepsilon=||\epsilon||\), and consider a  \(\delta-\)regular energy level  $H<b$. Then for \(I\in H_2^{-1}(\mathbb{H}_{2}^{R,\delta}(H))\backslash\mathcal{N}_{\rho}(I_{tan})\), for all \(\theta\),  for sufficiently small $\varepsilon_{}$ the return
map $\mathcal{F}_{\epsilon}:(I,\theta)\rightarrow(I',\theta')$ is symplectic,  $C^{r}$ smooth and $\varepsilon -C^r$ close to the unperturbed impact return map \(\mathcal{F}_{0}\) of Eq. (\ref{eq:int_impact_returnmap}). Namely, for all \((I,\theta)\) in this bounded domain, there exists \(\varepsilon_{0}(H,\delta,\rho)>0\), such that for all \(\varepsilon\in[0,\varepsilon_{0}(H,\delta,\rho))\): 
\begin{equation}
\mathcal{F}_{\epsilon}:
\begin{cases}
I'=I+\varepsilon f(I,\theta;\epsilon)\\
\theta'=\theta+\Theta(I,J(H,I))+\varepsilon g(I,\theta;\epsilon)
\end{cases}\label{eq:return_map_nearint}
\end{equation}
with $f,g$ $2\pi-$periodic in $\theta$,\ $f,g\in C^{r}$.
\end{thm}
\begin{proof} 
We first show that the perturbed return map  $\mathcal{F}_{\epsilon}:(I,\theta)\rightarrow(I',\theta')$
 can be decomposed to three maps:  
\begin{equation}
\mathcal{F}_{\epsilon}=\Phi_{\epsilon_r}^{[t^{*}_\epsilon,t^{**}_\epsilon]}\circ S_{\epsilon_{w}}\circ\Phi_{\epsilon_r}^{[0,t^{*}_\epsilon]}
\end{equation}
where $\Phi$ denotes the smooth Hamiltonian flow corresponding to system without the impact \(H(\cdot;\epsilon_{r},b=0)\) which governs the dynamics before and after impact,  $t^{*}_\epsilon$ and $t^{**}_\epsilon$ denote the time of impact with the wall and the time of return to the cross-section respectively and $S$ is the impact (gluing) map. The subscript indicates the dependence on the two different types of perturbations. 

Since \(\Phi_{\epsilon_r}^{[0,t]}\) is the $C^r$ smooth Hamiltonian flow corresponding to the Hamiltonian  \(H(\cdot;\epsilon_{r},b=0)\) which is  $C^r$ close to the unperturbed smooth Hamiltonian flow \(H(\cdot;\epsilon_{r}=0,b=0)\) for any finite \(t\), by considering \(I\) values only in the regular domain of the unperturbed flow, we insure that for sufficiently small \(\epsilon_{r}\) 
 the two smooth flows are \(\epsilon_{r}\)  $-C^r$ close on the finite time interval    \([0,t^{*}_0+1]\), where \(t^{*}_0\) denotes the finite unperturbed impact time. Since, away from \(t=\{0,t_{0}^{*}\}\), the unperturbed segment of the flow for the considered regular \(I \) values  is bounded away from \(\Sigma\) and from the section \(q_1=0\), for sufficiently small \(\epsilon_{r}\)  
the same statement holds for the perturbed smooth flow. Hence, for sufficiently small \(\epsilon_{w}\) there is no crossing of \(\Sigma\) or the wall coordinate occuring at \(t\)-values in the interior of the interval \((0,t_0^*)\). It follows that 
for sufficiently small \(\varepsilon \),
the perturbed first impact with the perturbed wall occurs before the trajectory returns to \(\Sigma\), is transverse, and the perturbed travel time \(t^{*}_\epsilon\)  is finite and \(\varepsilon\)   $-C^r$ close to   \(t^{*}_0\). Hence, for sufficiently small \(\epsilon_{w}\), the gluing map $S_{\epsilon_{w}}$ is $C^r$ smooth (for a $C^{r+1}$ smooth boundary \cite{chernov2006chaotic}) regardless of the form of the deformation or tilt of the wall, and, for sufficiently small  \(\varepsilon\), the composition   \(S_{\epsilon_{w}}\circ\Phi_{\epsilon_r}^{[0,t^{*}_\epsilon]}\) is  \(\varepsilon\)   $-C^r$ close to the unperturbed composition,    \(S_{0}\circ\Phi_{0}^{[0,t^{*}_0]}\). Namely, the perturbed trajectory just after the impact is  \(\varepsilon\)    $-C^r$ close to the unperturbed trajectory after impact. 

It follows that the perturbed trajectory which is propagated by the perturbed flow  $\Phi_{\epsilon_r}$ remains  \(\varepsilon\)   -  $C^r$ close to the unperturbed impact trajectory for finite times (e.g. past the unperturbed return time to the transverse cross section \(\Sigma\)),  and hence, by similar considerations as above, cannot collide with the wall or cross \(\Sigma\) at \(t\) values which are bounded away from \(t_0^{*}\) and \(t_0^{**}\) respectively. In particular, we obtain that the perturbed return time \(t^{**}_\epsilon\) is finite and  \(\varepsilon\)   $-C^r$ close to \(t^{**}_0\). 

Summarizing,  for \(I\in H_2^{-1}(\mathbb{H}_{2}^{R,\delta}(H))\backslash\mathcal{N}_{\rho}(I_{tan})\), for all \(\theta\), the return map \(\mathcal{F}_{\epsilon}\) to \(\Sigma\) includes a single transverse collision with the perturbed wall at \(t^{*}_\epsilon=t_{0}^{*}+\mathcal{O}(\varepsilon)\) and  thus the return map is of the form $\mathcal{F}_{\epsilon}=\Phi_{\epsilon_r}^{[t^{*}_\epsilon,t^{**}_\epsilon]}\circ S_{\epsilon_{w}}\circ\Phi_{\epsilon_r}^{[0,t^{*}_\epsilon]}$ and is  \(\varepsilon\)   $-C^r$ close to the unperturbed impact return map \(\mathcal{F}_{0}=\Phi_{0}^{[t^{*}_0,t^{**}_0]}\circ S_{0}\circ\Phi_{0}^{[0,t^{*}_0]}\)   given by    Eq. (\ref{eq:int_impact_returnmap}). 
\end{proof}

 Note that the return times and closeness results statements are non-uniform in \(H\). Establishing asymptotic results for large \(H\) requires more careful analysis of the bounds and constants appearing in the proof and will be deferred to later studies.

\begin{cor}\label{cor:nonrestori-persist}
For fixed \(H\) and \(\delta,\rho>0\), consider a circle which is bounded away from separatrices, tangencies and the non-twist set, i.e. \(I_{0}\)  belongs to the closed "good" set \(I_{0}\in H_2^{-1}(\mathbb{H}_{2}^{R,\delta}(H))\backslash(\mathcal{N}_{\rho}(I_{tan})\cup\mathcal{N}_{\rho}(\mathbf{I}_{NT}(H)))\equiv{}S_g(H,\delta,\rho)\). Furthermore, assume \(\Theta(I_{0},J(H,I_{0}))/2\pi\) is \((c,\nu)\)-Diophantine\begin{equation}
\mid\Theta(I_{0},J(H,I_{0}))-\frac{2\pi{}m}{n}\mid>cn^{-\nu-1}\  \forall{}m,n\in\mathbb{Z}\label{eq:diophantine}
\end{equation}
 where $1<\nu<\frac{1}{2}(r-1)$. Then, there exists  \(\varepsilon_{1}(H,\delta,\rho;c,\nu)\) such that for all  $\varepsilon<\varepsilon_{1}$  there exists a perturbed invariant circle  $(I_{\varepsilon}(\theta),\theta)$ with rotation number $\frac{\tilde{T}_{1}(J(H,I_0))}{T_{2}(I_0)}$ which is   $\varepsilon/c$ close to the unperturbed circle \(I=I_{0}\). Furthermore, the same result is valid for small \(c\) as long as $c$ is at least of \(\mathcal{O}(\sqrt{\varepsilon})\).
\end{cor}

\begin{proof}
From Theorem \ref{thm:near-int}, the map  (\ref{eq:return_map_nearint}) on \(S_g(H,\delta,\rho)\) is a \(C^{r}\) perturbation of an integrable twist map. It remains to show that the perturbed dynamics remain bounded away from tangency and separatrices - if this is shown, then the above corollary follows directly from KAM type results (see \cite{Arnold2007CelestialMechanics,meyer2008introduction,de2001tutorial}) applied to the map (\ref{eq:return_map_nearint}). Indeed, notice that \(
S_g(H,\delta,\rho)\subset{}S_g(H,\delta/2,\rho/2) 
\)
so  the upper bounds on \(\varepsilon \) of Theorem \ref{thm:near-int} for these sets satisfy $\varepsilon_0(H,\delta,\rho)>\varepsilon_0(H,\delta/2,\rho/2)$. Taking $\varepsilon<\varepsilon_0(H,\delta/2,\rho/2)$, insures that if $I_0\in{}S_g(H,\delta,\rho)$ then it is at least \(\Delta\) away from the boundary of \(S_g(H,\delta/2,\rho/2)\), where  \(\Delta=\min(\rho/4,K_{1}\delta,K_{2}\delta|\ln(2\delta)|)\) and  \(K_{1,2}(H) \) are some constants depending on the unperturbed rotation rates (see Eq. \ref{eq:omegalowerbnd}, \ref{eq:measIexcluded}). It follows that the map (\ref{eq:return_map_nearint}) is smooth in at least an $\mathcal{O}(\Delta)$ neighborhood for all $I_0\in{}S_g(H,\delta,\rho)$. Hence, by KAM theory, there exists \(\varepsilon^{*}(H,\delta,\rho;c,\nu)<\varepsilon_0(H,\delta/2,\rho/2),\) such that for all \(\varepsilon<\varepsilon^*\) near every $I_0\in{}S_g(H,\delta,\rho)$ with $c,\nu$-Diophantine \(\Theta(I_{0},J(H,I_{0}))/2\pi\), there exists a perturbed invariant curve with $I_\varepsilon(\theta)=I_0+\mathcal{O}(\frac{\varepsilon}{c})$ with the same rotation number as \(I_{0}\). Since $c$ is at least of \(\mathcal{O}(\sqrt{\varepsilon})\), there exists \(K>0\) such that  \(\frac{\varepsilon}{c}<K\sqrt{\varepsilon}\). Taking \(\varepsilon<\varepsilon_1(H,\delta,\rho;c,\nu)=\min(\varepsilon^{*}(H,\delta,\rho;c,\nu),\left(\frac{\Delta }{K}\right)^{2})  \) insures that \(K\sqrt{\varepsilon}<\Delta\), so the perturbed circle remains within the regular region $S_g(H,\delta/2,\rho/2)$ in which the map  (\ref{eq:return_map_nearint}) is smooth, as required.
\end{proof}

\begin{cor}
For sufficiently small $\varepsilon$, the complement to the set of all tori $I_0$  belonging to an energy surface \(H\) and satisfying the conditions of Corollary \ref{cor:nonrestori-persist} , namely the set of tori which do not necessarily  persist under \(\varepsilon\) perturbations is  of $\mathcal{O}(\sqrt{\varepsilon},\rho,\delta\ln\delta)$. 
\end{cor}
\begin{proof}
 The complement to the set  $S_g(H,\delta,\rho)$, namely the   $\delta-$ neighborhoods of separatrices and $\rho-$ neighborhoods of tangency and non-twist tori,  are of $\mathcal{O}(\rho,\delta\ln\delta)$, see proof of Theorem \ref{thm:(Smoothness)}. For     \(I_{0}\in S_g(H,\delta,\rho)\), by Corollary \ref{cor:nonrestori-persist}, KAM theorem \cite{Arnold2007CelestialMechanics} may be applied, hence in  $S_g(H,\delta,\rho)$  the complement set is the resonant set, and its measure is of $\mathcal{O}(\sqrt{\varepsilon})$.  \end{proof}

The destroyed, resonant tori correspond to rational values of the modified rotation number $\frac{\tilde{T}_{1}(J)}{T_{2}(I)}$. Notice that the impact causes a shift in the resonant frequencies.

The excluded sets (neighborhoods of separatrices, of the tangent torus and of the non-twist tori) correspond to a finite, discrete number of singular $I(H)$ values. As $\varepsilon\rightarrow0$, the size of these sets, which is controlled by $\delta,\rho$, can be taken to tend slowly to $0$ as well. The proof of corollary \ref{cor:nonrestori-persist} which utilizes KAM theory implies that in such a case \(\delta,\rho\) must be at least of  \(\mathcal{O}(\sqrt{\varepsilon})\). Finding the optimal power in \(\varepsilon\) is left for future studies. As the system is a 2 d.o.f system, this implies that the phase space can be divided into invariant regions of motion \cite{Arnold2007CelestialMechanics}.

There are two cases in which the form of the perturbed map for $I$ may be found. 
The first of which is the case of a perpendicular wall with an additional regular perturbation -  \(\epsilon_{w}=0\), so \( \varepsilon=\epsilon_{r}\). We introduce the following notation: let $z(t)=(q_{1}(t),p_{1}(t),q_{2}(t),p_{2}(t))$ and denote the impacting trajectory in the perturbed system by \(z_{\epsilon}(t)=\Phi_{\epsilon_r}^{[0,t]}z(0)\), for \(0\le{}t<t^{*}_\epsilon\), and \(z_{\epsilon}^{im}(t)=\Phi_{\epsilon_r}^{[t^{*}_\epsilon,t]}\circ S\circ\Phi_{\epsilon_r}^{[0,t^{*}_\epsilon]}z(0)\) for $t_\epsilon^{*}\le{}t\le{}t^{**}_\epsilon$. Denote similarly by $z_0^{im}(t)$ the trajectory in the unperturbed impact system. Finally, denote by $z_0^{sm}(t)=\Phi_0^{[0,t]}z(0)$ the smooth, unperturbed, non-impacting trajectory. From Theorem \ref{thm:near-int} above, we have:

\begin{cor}
\label{cor:z_eps-vs-z_0}Consider the settings of theorem \ref{thm:near-int}, with $\epsilon_{w}=0$ and $\epsilon_r$ sufficiently small. Let $\bar{t}_{min}=\min({t_0^{*},t^{*}_{\epsilon_r}})$, $\bar{t}_{max}=\max({t_0^{*},t^{*}_{\epsilon_r}})$. Then:
\begin{equation}
z_{\epsilon_{r}}(t)=
\begin{cases}
z_{0}^{im}(t)+\epsilon_{r}z_{1}^{im}(t)+\mathcal{O}(\epsilon_{r}^{2}) & t\in[0,\bar{t}_{min}]\cup[\bar{t}_{max},t_{\epsilon_{r}}^{**}] \\
\begin{cases}
z_{0}^{sm}(t)+\epsilon_{r}z_{1}^{sm}(t)+\mathcal{O}(\epsilon_{r}^{2}) & \bar{t}_{min}=t_{0}^{*} \\ (q_{1,0}^{sm}(t+\Delta t_{travel}),p_{1,0}^{sm}(t+\Delta t_{travel}),q_{2,0}^{sm}(t),p_{2,0}^{sm}(t))+\mathcal{O}(\epsilon_{r}) & \bar{t}_{max}=t_{0}^{*}
\end{cases}
&
t\in[\bar{t}_{min},\bar{t}_{max}]
\end{cases}
\end{equation}
Where $\Delta t_{travel}$ is as in Theorem \ref{thm:(Smoothness)},
calculated for the unperturbed impacting trajectory, and $z_{1}$ solves the first variational equation along the corresponding trajectory. \end{cor}

In other words, for $\varepsilon=\epsilon_r$, the perturbed trajectory $z_{\epsilon_{r}}(t)$ and the unperturbed trajectory $z_{0}^{im}(t)$ are $\mathcal{O}(\epsilon_{r})$ close except for an $\mathcal{O}(\epsilon_{r})$ time interval where one trajectory has already undergone impact and the other has not, in which case the perturbed trajectory can be approximated by the respective continuation of the unperturbed trajectory outside the billiard. 

\begin{figure}
\begin{centering}
\includegraphics[scale=0.5]{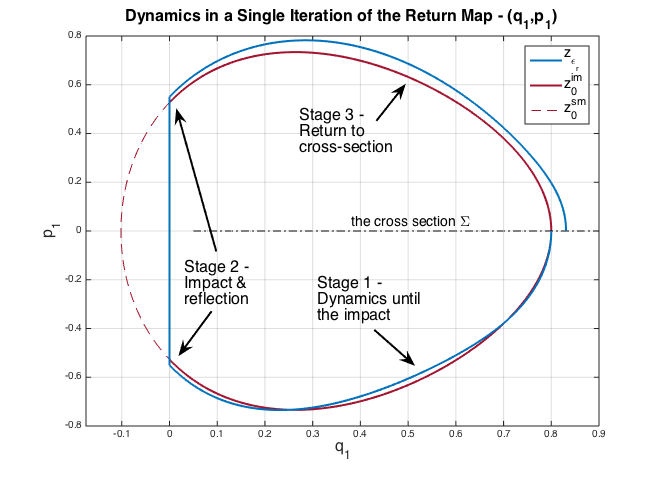}
\par\end{centering}

\protect\caption{\label{fig:singleit-dyn}The perturbed impacting trajectory $z_{\epsilon_{r}}$
(blue), unperturbed impacting trajectory $z_{0}^{im}$ (crimson) and unperturbed smooth
trajectory $z_{0}^{sm}$ (dashed crimson) during a single iteration
of the return map to the cross-section $\Sigma$, projected to the
$(q_{1},p_{1})$ phase space, for the case of a perpendicular wall and a small regular perturbation. The wall is at $q_{1}=0$.}
\end{figure}

\begin{thm}\label{thm:(Return-map-regpert)}
Consider a Hamiltonian $H$ of the form Eq. (\ref{eq:hgeneral}) with an S3B  integrable structure \(H_{int}\) and a regular wall position, with $\epsilon_w=0$.
Fix \(\delta>0, \rho>0\), and consider a  \(\delta-\)regular energy level  $H$. Then for \(I\in H_2^{-1}(\mathbb{H}_{2}^{R,\delta}(H))\backslash\mathcal{N}_{\rho}(I_{tan})\), for all \(\theta\),  for sufficiently small  $\epsilon_{r}$, the function \(f\) of  the change in $I$ in the return map (\ref{eq:return_map_nearint}) has the following form:
\begin{equation}\label{eq:return_map_perpnearint}
f(I,\theta;\epsilon_r)=\frac{1}{\omega_2(I)}\int_{0}^{\tilde T_1(J(I,H))}
\left(\frac{\partial{}V_r}{\partial{}q_2} \ p_{2}\right)_{z_0^{im}(t)}dt
+\mathcal{O}(\epsilon_r)
\end{equation}
\end{thm}

\begin{proof}
Consider the evolution in time of \(I\) under the perturbed system before, during and after impact (see Figure \ref{fig:singleit-dyn}). Before and after impact the motion is described by the smooth, near integrable Hamiltonian \(H(\cdot;\epsilon_{r},b=0)\) (Theorem \ref{thm:near-int} and Corollary \ref{cor:z_eps-vs-z_0}) and at impact, as the wall is perpendicular, \(I\) is unchanged. With no loss of generality, we consider the case  $\bar{t}_{min}=t_{0}^{*}<t_{\epsilon_{r}}^{*}$ (the other case may be similarly treated). 
 The evolution of \(I\) before the perturbed impact time \(t_{\epsilon_{r}}^{*}\) may be approximated by the evolution along the unperturbed trajectory until time \(t_{0}^{*}\):
\begin{equation}
\begin{split}
I^{*}&=I(t_{\epsilon_{r}}^{*})=I+\int_{0}^{t_{\epsilon_{r}}^{*}}\{I,H_{int}+\epsilon_{r}V_{r}\}\mid_{z_{\epsilon_{r}}(t)}dt=I+\int_{0}^{t_{0}^{*}}+\int_{t_{0}^{*}}^{t_{\epsilon_{r}}^{*}}\{I,\epsilon_{r}V_{r}\}\mid_{z_{\epsilon_{r}}(t)}dt \\
&= I+\epsilon_{r}\int_{0}^{t_{0}^{*}}\{I,V_{r}\}\mid_{z^{im}_{0}(t)}dt+\mathcal{O}(\epsilon_r^2) 
\end{split}
\end{equation}
where we used  $\{I,H_{int}\}=0$,  Theorem \ref{thm:near-int} and Corollary \ref{cor:z_eps-vs-z_0}. $I^{*}$ does not change at impact, and the evolution back to the cross-section $\Sigma$ after the impact may be calculated similarly:
\begin{equation}
\begin{split}
I'&=I(t_{\epsilon_{r}}^{**})=I(t_{\epsilon_{r}}^{*})+\int_{t_{\epsilon_{r}}^{*}}^{t_{\epsilon_{r}}^{**}}\{I,H_{int}+\epsilon_{r}V_{r}\}\mid_{z_{\epsilon_{r}}(t)}dt \\
&=I^{*}+\int_{t_{\epsilon_{r}}^{*}}^{t_{0}^{**}}\{I,\epsilon_{r}V_{r}\}\mid_{z_{\epsilon_{r}}(t)}dt+\int_{t_{0}^{**}}^{t_{\epsilon_{r}}^{**}}\{I,\epsilon_{r}V_{r}\}\mid_{z_{\epsilon_{r}}(t)}dt \\
&=I^{*}+\epsilon_{r}\int_{t_{\epsilon_{r}}^{*}}^{t_{0}^{**}}\{I,V_{r}\}\mid_{z^{im}_{0}(t)}dt+\mathcal{O}(\epsilon_r^{2})=I^{*}+\epsilon_{r}\int_{t_{0}^{*}}^{t_{0}^{**}}\{I,V_{r}\}\mid_{z_{0}^{im}(t)}dt+\mathcal{O}(\epsilon_r^2).
\end{split}
\end{equation}
Finally,  substituting $\{I,V_r\}=\frac{\partial{}V_r}{\partial{}\theta}$ and since $t_{0}^{**}=\tilde{T}_{1}(J)$ (see (\ref{eq:int_impact_returnmap})):
\begin{equation}
\begin{split}
I'&=I+\epsilon_r\int_{0}^{t_{0}^{**}}\frac{\partial{}V_r}{\partial{}\theta}\mid_{z_{0}^{im}(t)}dt+\mathcal{O}(\epsilon_r^2)=I+\epsilon_r\int_{0}^{t_{0}^{**}}\frac{\partial{}V_r}{\partial{}q_2}\cdot\frac{\partial{}q_2}{\partial{}\theta}\mid_{z_{0}^{im}(t)}dt+\mathcal{O}(\epsilon_r^2) \\
&=I+\epsilon_r\int_{0}^{\tilde{T}_{1}(J)}\frac{\partial{}V_r}{\partial{}q_2}\cdot\frac{\dot{q_2}}{\dot{\theta}}\mid_{z_{0}^{im}(t)}dt+\mathcal{O}(\epsilon_r^2)=I+\epsilon_r\frac{1}{\omega_2(I)}\int_{0}^{\tilde{T}_{1}(J)}\left(\frac{\partial{}V_r}{\partial{}q_2} \ p_{2}\right)_{z_0^{im}(t)}dt+\mathcal{O}(\epsilon_r^2)
\end{split}
\end{equation}
\end{proof}

The other case in which explicit form of the leading order term in $I$ can be written, is the case of a tilted, near perpendicular straight wall and a small regular perturbation ($q_1^w=\epsilon_{w}Q^w(q_2^w)=\epsilon_{w}q_2^w$,  $\epsilon_w,\epsilon_r$ small). In fact,  we show next that by rotating the coordinate system\ this case reduces to an example of the previous one. Consider first a tilted wall with no additional perturbation to the potential, so  $\epsilon_r=0$ and $q_1^w=\epsilon_{w}q_2^w$, $\epsilon_w$ small. The symplectic change of coordinates, of rotating the axes by
$\alpha=\arctan(\epsilon_w)$: 
\begin{equation}
\tilde{q}=Rq,\ \tilde{p}=Rp, R=\left(\begin{array}{cc}
\cos\alpha & -\sin\alpha\\
\sin\alpha & \cos\alpha
\end{array}\right)
\end{equation}
makes the wall perpendicular to the new $\tilde{q_{1}}$ axis, i.e.
$\tilde{Q}^w=Q^w(\tilde{q_2}^w)=0$. 
Substituting the new coordinates in the Hamiltonian, we obtain:
\begin{multline}
\tilde{H}(\tilde{q_{1}},\tilde{q_{2}},\tilde{p_{1}},\tilde{p_{2}})=H(\cos\alpha\cdot \tilde{q_{1}}+\sin\alpha\cdot \tilde{q_{2}},-\sin\alpha\cdot \tilde{q_{1}}+\cos\alpha\cdot \tilde{q_{2}},\cos\alpha\cdot \tilde{p_{1}}+\sin\alpha\cdot \tilde{p_{2}},-\sin\alpha\cdot \tilde{p_{1}}+\cos\alpha\cdot \tilde{p_{2}}) \\
=\frac{\tilde{p_{1}}^{2}}{2}+\frac{\tilde{p_{2}}^{2}}{2}+V_{1}(\cos\alpha\cdot \tilde{q_{1}}+\sin\alpha\cdot \tilde{q_{2}})+V_{2}(-\sin\alpha\cdot \tilde{q_{1}}+\cos\alpha\cdot \tilde{q_{2}})+b\cdot V_{b}(\tilde{q_{1}},\tilde{q_{2}})
\end{multline}

$V_{1,2}$ are $C^{r+1}$ functions and therefore can be expanded around $\tilde{q_{1}},\tilde{q_{2}}$ respectively:
\begin{equation}
\begin{split}
V_{1}(q_1)&=V_{1}(\cos\alpha\cdot \tilde{q_{1}}+\sin\alpha\cdot \tilde{q_{2}})=V_{1}(\tilde{q_{1}})+(\tilde{q_{1}}-(\cos\alpha\cdot \tilde{q_{1}}+\sin\alpha\cdot \tilde{q_{2}}))\cdot V_1'(\tilde{q_{1}}) \\
&+(\tilde{q_{1}}-(\cos\alpha\cdot \tilde{q_{1}}+\sin\alpha\cdot \tilde{q_{2}}))\cdot h_{1}(\cos\alpha\cdot \tilde{q_{1}}+\sin\alpha\cdot \tilde{q_{2}}) \\
V_{2}(q_2)&=V_2(-\sin\alpha\cdot \tilde{q_{1}}+\cos\alpha\cdot \tilde{q_{2}})=V_{2}(\tilde{q_{2}})+(\tilde{q_{2}}-(-\sin\alpha\cdot \tilde{q_{1}}+\cos\alpha\cdot \tilde{q_{2}}))\cdot V_{2}'(\tilde{q_{2}}) \\
&+(\tilde{q_{2}}-(-\sin\alpha\cdot \tilde{q_{1}}+\cos\alpha\cdot \tilde{q_{2}}))\cdot h_{2}(-\sin\alpha\cdot \tilde{q_{1}}+\cos\alpha\cdot \tilde{q_{2}})
\end{split}
\label{eq:rotation_expansion}
\end{equation}
Where $h_{1},h_{2}\rightarrow0$ as $\epsilon_w\rightarrow0$. 
For $\epsilon_w$ small, the trigonometric functions can also be expanded:
\begin{equation}
\begin{split}
V_{1}(\cos\alpha\cdot \tilde{q_{1}}+\sin\alpha\cdot \tilde{q_{2}})&=V_{1}(\tilde{q_{1}})+\epsilon_w V_{1,rem}(\tilde{q_{1}},\tilde{q_{2}};\epsilon_w) \\
V_{2}(-\sin\alpha\cdot \tilde{q_{1}}+\cos\alpha\cdot \tilde{q_{2}})&=V_{2}(\tilde{q_{2}})+\epsilon_w V_{2,rem}(\tilde{q_{1}},\tilde{q_{2}};\epsilon_w)\\ 
V_{rem}(\tilde{q_{1}},\tilde{q_{2}};\epsilon_w)&=V_{1,rem}(\tilde{q_{1}},\tilde{q_{2}};\epsilon_w)+V_{2,rem}(\tilde{q_{1}},\tilde{q_{2}};\epsilon_w)
\end{split}
\end{equation}
where \(V_{rem}(\cdot)\) is \(C^{r}\) and  in particular bounded on the perturbed energy surface (though non-uniformly in \(H\)).
The form of the Hamiltonian in the new rotated coordinates is
\begin{equation}
\tilde H(\tilde{q_{1}},\tilde{q_{2}},\tilde{p_{1}},\tilde{p_{2}})=\frac{\tilde{p_{1}}^{2}}{2}+\frac{\tilde{p_{2}}^{2}}{2}+V_{1}(\tilde{q_{1}})+V_{2}(\tilde{q_{2}})+\epsilon_w V_{rem}(\tilde{q_{1}},\tilde{q_{2}};\epsilon_w)+b\cdot V_{b}(\tilde{q_{1}};0)
\end{equation}
where, using the expansion in (\ref{eq:rotation_expansion}), we have:
\begin{equation}
V_{rem}=-\tilde{q_{2}}\cdot V_1'(\tilde{q_{1}})+\tilde{q_{1}}\cdot V_2'(\tilde{q_{2}})+\mathcal{O}(\epsilon_w)
\end{equation}
namely, the integrable part of \(\tilde H\), in the rotated coordinates, is exactly \(H_{int}\). Notice that due to the expansion, the smoothness of the leading order perturbation term is reduced by one. We establish:
\begin{cor}
\label{cor:nearperpint-equiv-perpnearint} For ${H},\epsilon_w$ and initial conditions $(\tilde{I},\tilde{\theta})$ which satisfy the assumptions of Theorem \ref{thm:near-int} with $r>4$, an impact by a near perpendicular straight wall is equivalent to the system with impact with a perpendicular wall and a small, regular
perturbation. Moreover, the form of the change in \(\tilde{I}\) due to the wall tilt becomes  (see Theorem \ref{thm:(Return-map-regpert)}): 
\begin{equation}\label{eq:return_map_nearperpint}
f(\tilde{I},\tilde{\theta};\epsilon_w)=\frac{1}{\omega_2(\tilde{I})}\int_{0}^{\tilde{T}_{1}(\tilde{J})}\bigg([-V_1'(\tilde{q_{1}})+\tilde{q_{1}}\cdot V_2''(\tilde{q_{2}})] \ \tilde{p_{2}}\bigg)_{\tilde{z_0}^{im}(t)}dt +\mathcal{O}(\epsilon_w)
\end{equation}
\end{cor}

Similarly, when both $\epsilon_w, \epsilon_r\neq0$ and are of the same order, one  finds that the Hamiltonian in the rotated coordinates corresponds to a system with a perpendicular wall and a small regular perturbation, comprised of a rotation term $\epsilon_wV_{rem}(\tilde{q_1},\tilde{q_2})$ and the original regular perturbation $\epsilon_r\tilde{V}_r(\tilde{q_1},\tilde{q_2};\epsilon_r,\epsilon_w)$. Using a similar expansion to (\ref{eq:rotation_expansion}), we have:
\begin{equation}\label{eq:taylor_regpertrotated}
V_r(q_1,q_2;\epsilon_r)=V_r(\tilde{q_1},\tilde{q_2};\epsilon_r)+\epsilon_w\Big(-\tilde{q_2}\frac{\partial{}V_r}{\partial{}q_1}(\tilde{q_1},\tilde{q_2})+\tilde{q_1}\frac{\partial{}V_r}{\partial{}q_2}(\tilde{q_1},\tilde{q_2})\Big)+\mathcal{O}(\epsilon_w^2):=\tilde{V}_r(\tilde{q_1},\tilde{q_2};\epsilon_r,\epsilon_w)
\end{equation}
Assuming that  $\epsilon_r=c_r\epsilon,\ \epsilon_w=c_w\epsilon$, the form of the change in  \(\tilde{I}\) becomes (see Theorem \ref{thm:(Return-map-regpert)}, equations (\ref{eq:rotation_expansion},\ref{eq:taylor_regpertrotated})):
\begin{equation}\label{eq:return_map_nearperpnearint}
f(\tilde{I},\tilde{\theta};\epsilon)=\frac{1}{\omega_2(\tilde{I})}\int_{0}^{\tilde{T}_{1}(\tilde{J})}\bigg(\Big[c_w[-V_1'(\tilde{q_{1}})+\tilde{q_{1}}\cdot V_2''(\tilde{q_{2}})]+c_r\frac{\partial{}V_r}{\partial{}\tilde{q_2}}\Big] \ \tilde{p_{2}}\bigg)_{\tilde{z_0}^{im}(t)}dt +\mathcal{O}(\epsilon).
\end{equation}

\subsubsection*{Soft impacts}
 For physical setups in which bodies at close range experience strong repulsion forces (e.g. the repelling forces between two colliding atoms) \cite{rapoport2007approximating,rom2012billiards,lerman2012saddle,RK2014smooth}, a better model for the strong repulsion than the singular hard-wall billiard potential is  a smooth steep potential. Hence, consider Hamiltonian systems similar to those discussed above, where the hard billiard is replaced by a smooth potential whose softness is controlled by a small parameter $\epsilon_b$:
\begin{equation}
H=H(\cdot;\epsilon_{r},\epsilon_{w},\epsilon_{b},q^{w},b)=H_{int}(q_{1},p_{1},q_{2},p_{2})+\epsilon_{r}V_{r}(q_{1},q_{2})+ b\cdot{}V_{b}(q_{};\epsilon_{w},\epsilon_{b})\label{eq:hsoft}
\end{equation}
As $\epsilon_{b}\rightarrow{}0$, the smooth ($C^{r+1}$) billiard potential $V_b(\cdot,\epsilon_b)$ becomes steeper at the wall (as \(q-q^{w}\rightarrow0^+\)) and approaches the singular hard wall limit. For example, one can choose (see  \cite{rapoport2007approximating,RK2014smooth} for additional examples):\begin{equation}
V_{b,poly}(q_{};\epsilon_{w},\epsilon_{b})=\frac{\epsilon_{b}}{q_1-\epsilon_{w}Q^{w}(q_{2};\epsilon_{w})^{}}
\label{eq:infinitebarrier}\end{equation}
\begin{equation}
V_{b,exp}(q_{};\epsilon_{w},\epsilon_{b})=\exp\left(-\frac{q_1-\epsilon_{w}Q^{w}(q_{2};\epsilon_{w})}
{\epsilon_{b}}\right)\label{eq:finitebarrier}
\end{equation} Notice that in particular, on the wall, \(\lim _{q\rightarrow q^{w}}b\cdot{}V_{b}(q;\epsilon_{w},\epsilon_{b})\geq b\) (this limit, which corresponds to the "barrier height", is infinite for the potential \(V_{b,poly}\) and finite for \(V_{b,exp}\)). It has been shown \cite{rapoport2007approximating,RK2014smooth} that under some natural conditions on $V_b$, for trajectory segments that are bounded away from tangencies and have energies which are not too large (so they cannot cross the boundary), the smooth, soft impact flow and the piecewise-smooth, hard impact flow are $C^r$ close on a section bounded away from the impact boundary. The detailed conditions of  \cite{RK2014smooth} and their realization in the context of the current setup are included, for completeness, in appendix B. Then, the results in \cite{RK2014smooth} can be used to prove a  somewhat weaker version of Theorem \ref{thm:near-int} that applies to the soft impact case (in particular, unless \(\epsilon_{b}\) is taken to be very small, the form of the perturbed return map also depends on the errors gathered by the singular perturbation term, see corollary \ref{cor:fformsoft}): 

\begin{thm}\label{thm:softnearintegrable}
Consider a Hamiltonian $H$ of the form Eq.
(\ref{eq:hsoft}) with an S3B  integrable structure \(H_{int}\),  a regular wall position, and a soft billiard potential $V_b$ satisfying conditions I-IV (see appendix B).   Fix \(\delta>0, \rho>0\), let ${\epsilon}=(\epsilon_{r},\epsilon_w,\epsilon_b)$ and \(\varepsilon=||\epsilon||\), and consider a  \(\delta-\)regular energy level $H$ satisfying $H<H_{max}(b)$ (see appendix). Then for \(I\in H_2^{-1}(\mathbb{H}_{2}^{R,\delta}(H))\backslash\mathcal{N}_{\rho}(I_{tan})\), for all \(\theta\),  for sufficiently small $\varepsilon_{}$ the return
map $\mathcal{F}_{\epsilon}:(I,\theta)\rightarrow(I',\theta')$ is symplectic,  $C^{r}$ smooth and $C^k$ close to the unperturbed impact return map \(\mathcal{F}_{0}\) of Eq. (\ref{eq:int_impact_returnmap}) for any $k\leq{}r$. Namely, for all $(I,\theta)$ in this bounded domain, there exists $\varepsilon_b(H,\delta,\rho)>0$ such that for all $\varepsilon\in[0,\varepsilon_b(H,\delta,\rho))$, $\mathcal{F}_\epsilon=\mathcal{F}_0+\mathit{o}_{C^k}(1)$.
\end{thm}

\begin{proof}
Symplecticity and smoothness of the soft impact flow, and hence the map, are immediate. Since the transverse section \(\Sigma\) is bounded away from the wall, and since we consider orbits which are bounded away from being tangent, the \(C^{k}\) closeness of \(\mathcal{F}_{(\epsilon_{r},\epsilon_w,\epsilon_b)}\) and \(\mathcal{F}_{(\epsilon_{r},\epsilon_w,\epsilon_b=0^+)}\) follows from Theorem 1 in \cite{RK2014smooth} (see appendix, where the conditions of  Theorem 1 in \cite{RK2014smooth} are shown to be satisfied here, and the bounds on  \(H\) are shown to guarantee that for sufficiently small \(\varepsilon\) particles cannot cross the wall). The  $C^r$ closeness of  \(\mathcal{F}_{(\epsilon_{r},\epsilon_w,\epsilon_b=0^+)}\)  to \(\mathcal{F}_{0}\) for sufficiently small $\epsilon_r,\epsilon_w$ follows from Theorem \ref{thm:near-int}. \end{proof}

\begin{rem*}
Notice that the approximation of the near integrable map by the integrable one is weaker here, as the error is $\mathit{o}(1)$ in $\epsilon_b$, versus the $\mathcal{O}(\varepsilon)$ error in Theorem \ref{thm:near-int}. This is due to the singular nature of the soft billiard perturbation, as opposed to the regular perturbations of the Hamiltonian structure or the vertical wall shape. Error estimates for the hard billiard case have been calculated in \cite{rapoport2007approximating} for some specific forms of $V_b$. These estimates may be extended to the soft impact case and used to derive an explicit formula for the first order approximation term to the soft impact return map. The exact formulation is left for future works. However, the existence of such estimates implies that there exists $\epsilon_b$ sufficiently small such that the error remains $\mathcal{O}(\varepsilon)$, i.e.
\end{rem*}

\begin{cor}\label{cor:fformsoft}
There exists $\epsilon_{b,k}(\epsilon_r,\epsilon_w)$ such that for all $\epsilon_b<\epsilon_{b,k}$, under the conditions of Theorem \ref{thm:softnearintegrable}, the soft impact return map $\mathcal{F}_\epsilon$ is $\varepsilon-C^k$ close to $\mathcal{F}_0$ and in the special calculable cases  the first order term in $\varepsilon$ of the soft impact return map takes the corresponding forms (\ref{eq:return_map_perpnearint}), (\ref{eq:return_map_nearperpint}) or (\ref{eq:return_map_nearperpnearint}).
\end{cor}

For example, we conjecture that if for a given soft potential form the error estimate for $C^k$ closeness as in \cite{rapoport2007approximating} is of $\mathcal{O}(\sqrt[k+2]{\epsilon_b})$, then for $\epsilon_b\leq\mathcal{O}(\epsilon_r^{k+2},\epsilon_w^{k+2})$ the overall error would be of $\mathcal{O}(\varepsilon)$ as required.

\section{Example  } \label{sec:Examples}

Consider the Hamiltonian $H_{int}=\frac{p_{1}^{2}}{2}+\frac{p_{2}^{2}}{2}-\frac{\lambda^{2}}{2}\cdot(q_{1}-q_{1s})^{2}+\frac{1}{4}\cdot(q_{1}-q_{1s})^{4}+\frac{\omega^{2}}{2}\cdot(q_{2}-q_{2s})^{2}$. In the $(q_{1},p_{1})$ plane the Hamiltonian flow has
a saddle point $(q_{1s},0)$ and a separatrix loop which encircles two symmetric centers   - the undamped Duffing
oscillator \cite{meyer2008introduction}; In $(q_{2},p_{2})$ there is a single linear center (see Figure \ref{fig:Duffing-ps}). The period in the $(q_{1},p_{1})$ plane is piecewise monotone; it becomes infinite at the separatrix  $(H_{1}=0)$, where it
reverses its  direction of monotonicity: \(T_{1}'(H_{1})\cdot{}\text{sign} (H_{1})<0\)
(see Figure \ref{fig:twistnoimpact}).  The period around the linear center is fixed: $T_{2}(H_{2})=\frac{2\pi}{\omega}$.  In the regular regions away from the  separatrix, action-angle variables can be defined
 and $H_{int}=H_{1}(J)+\omega I$. The periods' relation $2\pi\frac{T_1(H_1)}{T_2(H_2)}=\omega{}T_{1}(H-\omega I)$ in these regions is monotone, and so the smooth return map (\ref{eq:int_returnmap}) for $H_{int}$ is a twist map.
Consider now the impact system (\ref{eq:hgeneral}) with $H_{int}$ as above and $\epsilon_{r}=\epsilon_w=0$. The parameters $q_{1s},\lambda$ can be chosen such that the wall location is either inside or outside the separatrix loop (see  \cite{pnueli2016thesis}
for the different parameter ranges and the list of singular cases).
To demonstrate the results of sections \ref{sec:The-Integrable-System} and \ref{sec:Near-integrability-results}, we consider here  two regular cases: a) the tangent level set encircles the separatrix from outside, and b) the tangent level set is inside the separatrix, to the left of the left center point, see Figure \ref{fig:impact+duffing}. We show that in the former case the non-twist set remains empty whereas in the latter case there is an impacting non-twist torus (see Figure \ref{fig:twist-insidevsoutsidesep}).
\begin{figure}[h]
\begin{centering}
\includegraphics[scale=0.4]{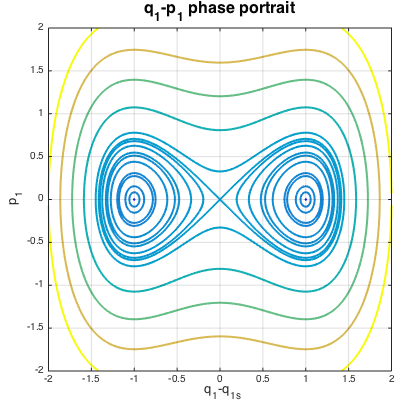}\includegraphics[scale=0.4]{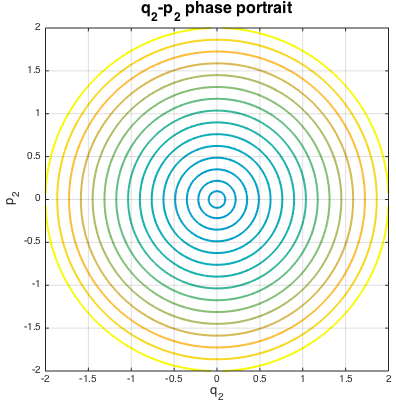}
\par\end{centering}
\protect\caption{\label{fig:Duffing-ps}Energy level lines in the phase space $(q_{1},p_{1})$ (left), $(q_{2},p_{2})$ (right)}
\end{figure}

\begin{figure}
\begin{centering}
\includegraphics[scale=0.35]{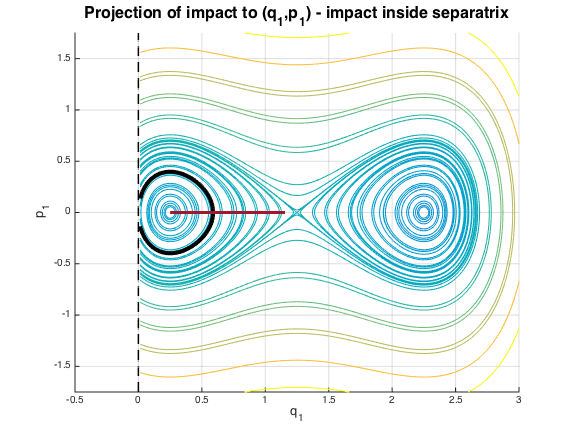}\includegraphics[scale=0.37]{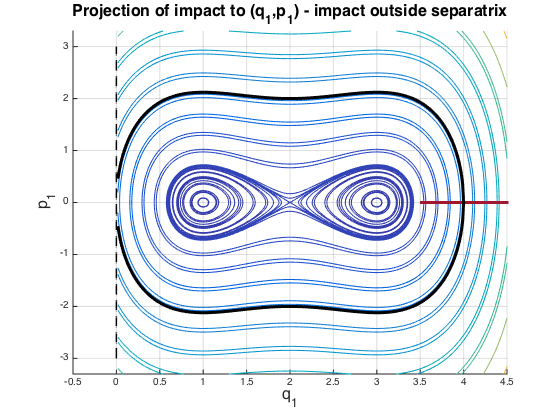}
\par\end{centering}
\protect\caption{\label{fig:impact+duffing}The separatrix and the wall locations. Here, $\lambda=1$. The wall,  located at $q^w_{1}=0$, 
can be either inside (left, \(q_{1s}=1.25)\)  or outside (right,  \(q_{1s}=2)\)   the separatrix.
The tangent level line corresponding to $J_{tan}$ is indicated in black. The cross-section $\Sigma$ of the return map is depicted in crimson, and is defined on the impacting branch.}
\end{figure}

\begin{figure}
\begin{centering}
\includegraphics[scale=0.4]{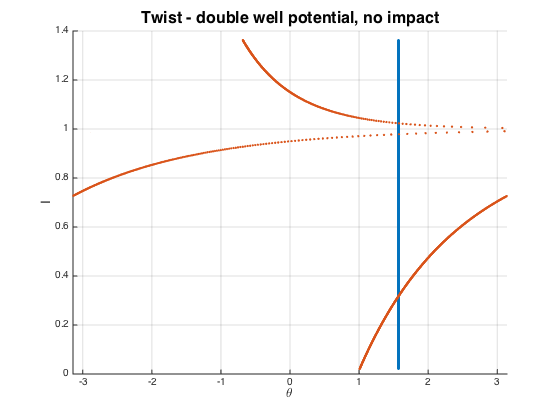}
\par\end{centering}
\protect\caption{\label{fig:twistnoimpact}The twist in the system described
by $H_{int}$, on the energy surface $H=1$. In blue is a line of initial $(I,\theta)$ values, and
in red - the corresponding $(I'=I,\theta')$ values following a single iteration
of the integrable return map (wrapped in $2\pi$). Here $I=1$ corresponds to
$H_{1}=0$. Due to the singularity at the separatrix, the twist changes
direction between energies inside and outside the separatrix. In either
case, there is no non-twist torus, due to the monotonicity of $T_{1}$.}
\end{figure}

In fact, one can show that for any regular wall position of the first type (or, respectively, of the second type) the non-twist set remains empty (respectively, has at least one non-twist torus). Indeed, this follows from the fact that for the impact system, on the regular set, \(\Theta(I)=\omega\tilde T_{1}(H-\omega{}I)=\omega T_{1}(H-\omega{}I)-\omega\Delta{}t_{travel}(H-\omega{}I)\), so \(\Theta'(I)=-\omega{}^{2}(T_{1}'(H-\omega{}I)-\Delta{}t_{travel}'(H-\omega{}I))\).
For \(I= (H\pm\delta)/\omega\), the first term approaches  \(\pm1/\delta\) as $\delta\rightarrow0$ (there \(H_{1}\) is in the \(\delta\) neighborhood of the separatrix). However, for \(I=( H-V_{1}(0)-\rho)/\omega\) (where \(H_{1}\) is larger by \(\rho\) from the tangent energy leaf, \(V_{1}(0)\)), the second term approaches  \(1/\sqrt{\rho}\) as $\rho\rightarrow0$.  For regular wall position the tangent level set and the separatrix are bounded away from each other, and thus it follows that if the tangency occurs inside the separatrix (\(V_{1}(0)<0\)) then, for sufficiently small \(\delta,\rho\),\ the rotation function derivative, \(\Theta'(I)\), must change sign over the interval \(I\in[(H+\delta)/\omega , \left(H-V_{1}(0)-\rho\right)/\omega]\subset S_g(H,\delta,\rho/\omega)\), and hence there exists at least one non-twist torus in the good set.    On the other hand,   if \(V_{1}(0)>0\) then \(\Theta'(I)<0\) for all \(I<H/\omega \) and thus the non-twist set remains empty even with impact.
See Figure \ref{fig:twist-insidevsoutsidesep} for illustration. 

\begin{figure}
\begin{centering}
\includegraphics[scale=0.45]{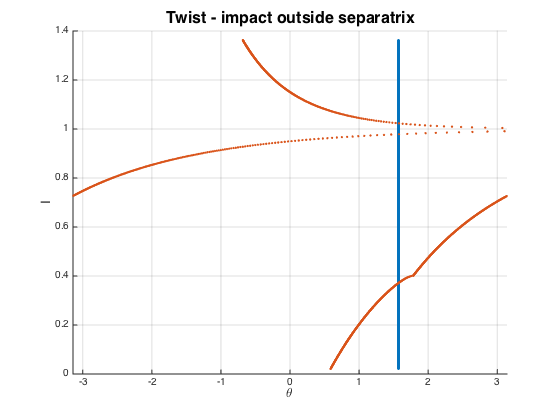}\includegraphics[scale=0.45]{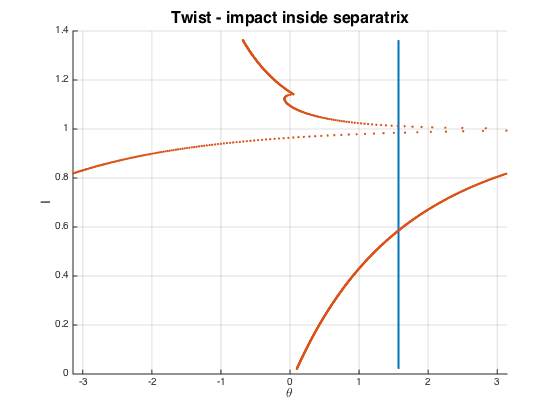}
\par\end{centering}

\protect\caption{\label{fig:twist-insidevsoutsidesep} The twist in $\theta$ when
impact is outside (left) or inside (right) of the separatrix. The points
of non-smoothness correspond to $I_{tan}$, in
which $\Delta t_{travel}$ is continuous but not smooth. $I$ values
below this value correspond to impacting trajectories. As can be seen
to the left, when the impact is outside the separatrix ($V_1(0)>0$) the impact
``contributes'' to the same direction of the original, non-impacting
twist. To the right, twist monotonicity is destroyed when impact is
inside the separatrix ($V_1(0)<0$) and a single non-twist torus is created at an intermediate point $I_{NT}\in{}(1,I_{tan})$.}
\end{figure}

\subsubsection*{Near integrability results}

Figures \ref{fig:nearintresults}-\ref{fig:I-EMBD} demonstrate numerically the
near-integrability results described in section \ref{sec:Near-integrability-results}, and in particular the equivalence between the perpendicular and near perpendicular cases. Figure \ref{fig:nearintresults} depicts the dynamics of the return map in the $(\theta,I)$ plane. Examined are the cases of a near perpendicular, straight wall with an underlying integrable structure, a perpendicular wall with underlying near integrable structure, and the near perpendicular, near integrable combination. In all three cases, near integrable behavior in the form of KAM tori and resonances can be seen in the regions bounded away from tangency and the separatrix. Identification of these regions is made easily using the Impact Energy-Momentum Bifurcation Diagrams in Figure  \ref{fig:I-EMBD} (see below). For impacting trajectories the similarities between all three cases  are  evident. For non-impacting trajectories, integrable behavior is seen at the top figure ($\epsilon_w=0.01$, $\epsilon_{r}=0$) whereas, naturally, the remaining cases ($\epsilon_r\neq{}0$) exhibit near integrable behavior even when the trajectories do not hit the wall\footnote{Near the separatrix  the map (\ref{eq:return_map_nearint}) is not well defined: the same $(\theta,I)$ values may correspond to two different sections in the \((q_{1},p_1)\) plane. One needs to use the separatrix map to obtain well defined sections there. Since the separatrix is not studied here, yet we want to present the global behavior, we do extend the marked section across the separatrix and ignore for now the observed artificial multiplicity which appears for trajectories that cross the separatrix.
}. 

In Figure \ref{fig:I-EMBD}, the same dynamics are depicted in the $(H_{int},I)$ plane, using an Impact Energy-Momentum Bifurcation Diagram,  providing insights about the structure of the flow at different energy values; The classical Energy-Momentum Bifurcation Diagram (EMBD) \cite{lerman1998integrable,Arnold2007CelestialMechanics} for the smooth Hamiltonian is, in our case, a plot in the $(H,I)$ space, where $H$ is the energy of the integrable system and $I$ is the action variable in the $(q_2,p_2)$ phase space.
In this plot the regions of allowed motion are shaded grey, and the curves corresponding to the $(H_{int},I)$ values on singular level sets of the system are depicted as dashed lines (respectively, solid lines) for singular level sets that include normally hyperbolic (respectively, normally elliptic)  circles. Together with either Fomenko graphs or indicators of the number of Liouville leaves in each region \cite{fomenko2004integrable, Arnold2007CelestialMechanics}, such plots help to classify the dynamics on different energy surfaces.

Here we introduce a new variant to this representation, the Impact-EMBD, in which we add the projection of the conditions of impact (blue) and tangency (green) into the EMBD. When the wall is perpendicular, this projection results in a line which corresponds to tangent tori, and which separates between impacting and non-impacting trajectories. When the wall is not perpendicular, due to the breaking of the symmetry, the condition for tangency projects onto the I-EMBD as a 2-dimensional zone. While in the symmetric case each point on the tangency line corresponded to tori on which all initial conditions achieved tangency at first collision, in the non-symmetric case this is satisfied by only a finite (see \cite{pnueli2016thesis}) number of points on each torus in the tangency zone. For the non-perpendicular wall the minimal energy for impact again coincides with the minimal energy for a possible tangency. By projecting the dynamics into the I-EMBD we achieve a classification of the different types of trajectories in relation to the impact and internal phase space structure. These behaviors are then demonstrated in the $(\theta,I)$ plane (notice that the vertical axis in both projections corresponds to $I$ values, which, together with fixing the total energy, allows for straightforward inference between the two different projections).

\begin{figure}[h]
\begin{centering}
\includegraphics[scale=0.45]{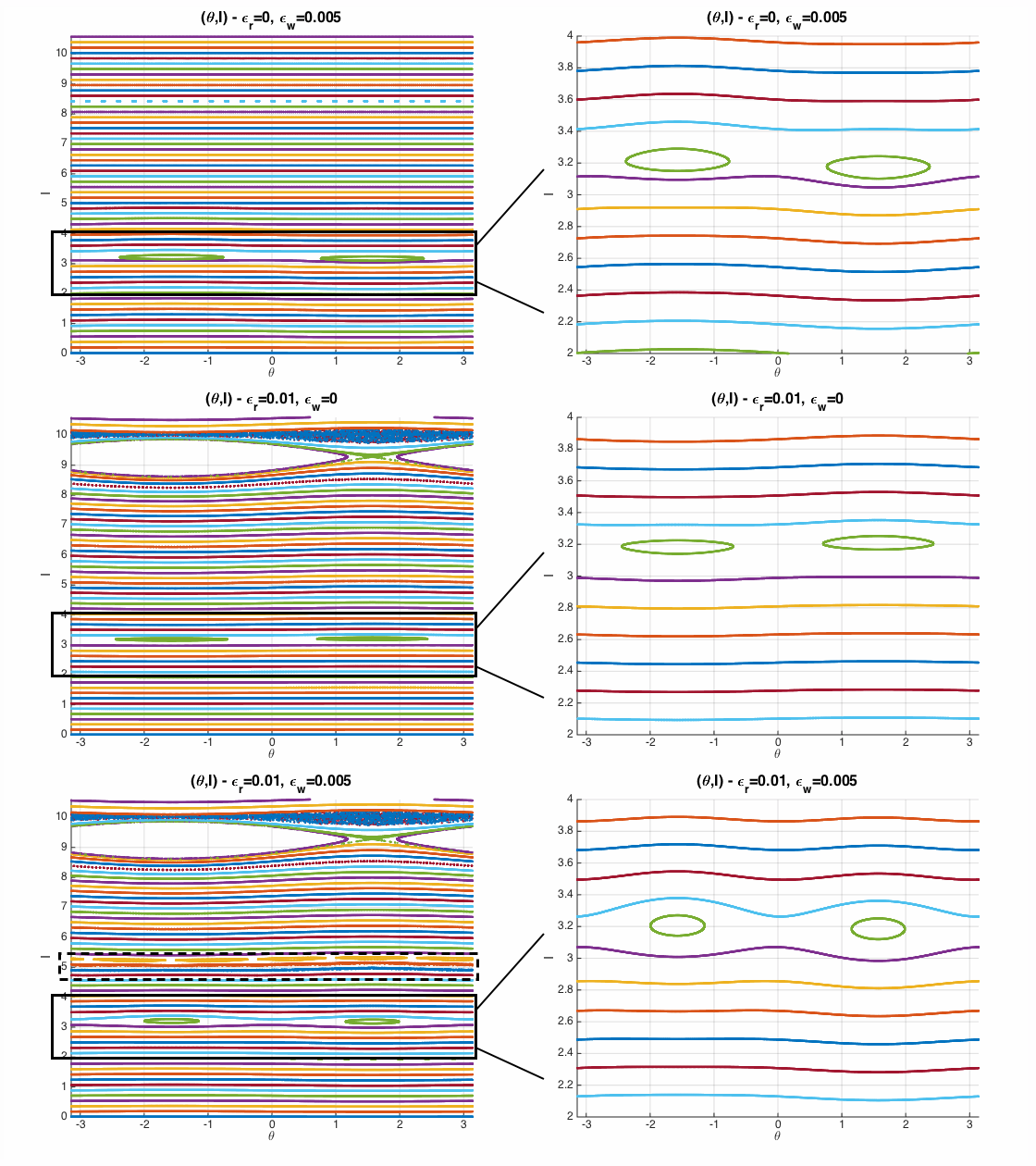}
\par\end{centering}

\protect\caption{\label{fig:nearintresults}Poincar\'{e} return map $(\theta',I')$ for the
following cases: (top) $q_1^w=\epsilon_wq_2^w, \epsilon_w=0.01$, $\epsilon_{r}=0$;
(middle) $\epsilon_w=0$, $\epsilon_{r}=0.005$; (bottom) $q_1^w=\epsilon_wq_2^w, \epsilon_w=0.01$, $\epsilon_{r}=0.005$. Initial conditions
for all three figures are the same. To the left, the entire possible range of $I$ values is depicted. To the right, a zoom on a region away from the separatrix
and from tangency is shown. KAM tori and resonances can be seen, as well as
the similarity between the three different settings. The  distinction between impacting and non-impacting $I$ values can be easily made by comparison with Figure \ref{fig:I-EMBD}.}

\end{figure}

\begin{figure}[h]
\begin{centering}
\includegraphics[scale=0.3]{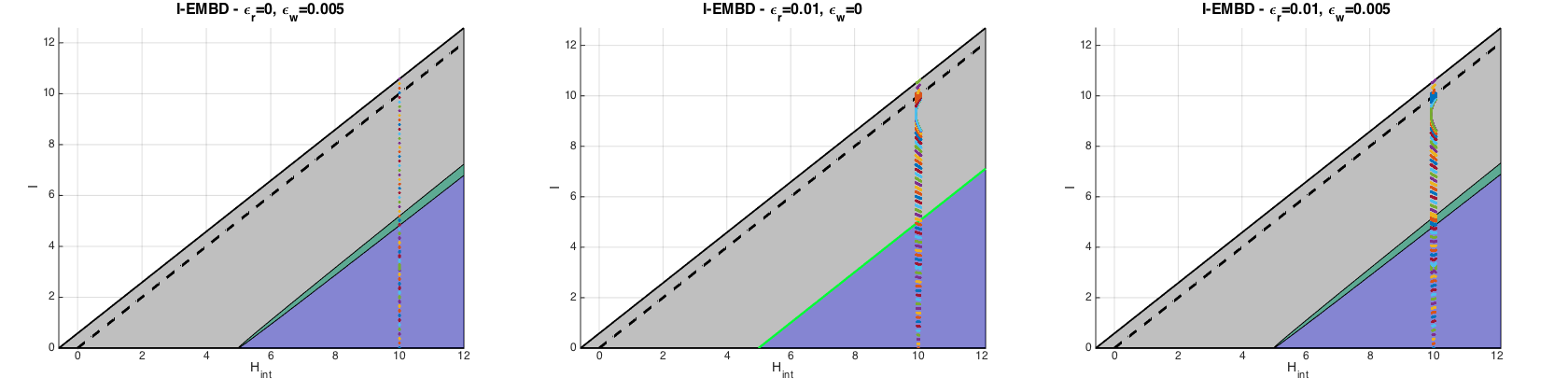}
\par\end{centering}

\protect\caption{\label{fig:I-EMBD}I-EMBD $(H_{int},I)$ for the
following cases: (left) $q_1^w=\epsilon_wq_2^w, \epsilon_w=0.01$, $\epsilon_{r}=0$;
(middle) $\epsilon_w=0$, $\epsilon_{r}=0.005$; (right) $q_1^w=\epsilon_wq_2^w, \epsilon_w=0.01$, $\epsilon_{r}=0.005$. Notice that in the cases of a small tilt (left and right), tangency is projected as a zone, whereas for the vertical wall it is projected as a line. The return map values depicted in Figure \ref{fig:nearintresults} are projected here into the I-EMBD.}

\end{figure}

\section{Discussion} \label{sec:Conclusions}
Near integrability results for a class of 2 d.o.f separable mechanical impact systems with a single wall were derived. When the wall conserves the symmetry of the integrable system - here, the separability - the system remains integrable. In particular, local sections allow to define Poincar\'{e} return maps that are smooth and satisfy the twist condition. We proved that breaking the separability of the system by the addition of a small regular perturbation, a small perturbation of the wall, making the wall soft, or a combination of all these effects together, may destroy the integrability of the return map, yet the map remains near integrable  for a large portion of the phase space (Theorems \ref{thm:near-int}, \ref{thm:(Return-map-regpert)} and \ref{thm:softnearintegrable}). For the case of a small regular perturbation and a slightly tilted, straight wall, an explicit form of the first order term in the perturbed return map was derived, a form which applies also to the soft impact formulation in the limit of very steep potential. The correction terms which arise from the steep potential part could be possibly derived as well (see \cite{rapoport2007approximating}). 

The dynamics near singularities of the impact system - separatrices and tangencies - are yet to be explored, as is the limit of large energy values.  Away from tangency, the dynamics near the separatrix are expected to exhibit the usual separatrix splitting and homoclinic chaos, similar to the smooth case.
The  near tangent dynamics are expected to produce more exotic behavior, as is demonstrated in Figure  \ref{fig:neartangent}. Notably, some aspects of this behavior have been explored by Neishtadt in \cite{neishtadt2008jump}, for a 1.5 d.o.f system with slow-fast dynamics. The system (\ref{eq:hgeneral})  may be reduced to such a system when $\omega_{1}(J)\gg\omega_{2}(I)$; Indeed,  let us denote $\frac{T_{1}(J)}{T_{2}(I)}=\delta$, where $\delta>0$ and small,  define the slow time variable $\tau=\delta t$ and symbolically denote the slow variables
\begin{equation}
q_{2}=q_{2}(\tau),\ p_{2}=p_{2}(\tau)
\end{equation}
Since  $q_1^w=\epsilon_wQ(q_2^w)$  
the collision point with the wall varies slowly with the evolution
of $q_{2}$ - $q_{1}^{*}=\epsilon_wQ(q^{*}_{2}(\tau))$. Similarly,
the perturbation $\epsilon_{r}V_{r}(q_{1},q_{2})$ changes slowly
in time - $\epsilon_{r}V_{r}(q_{1};\tau)$. The slow-fast system can
therefore be effectively reduced to a 1.5 degrees of freedom system
with a slowly varying potential and a slowly moving wall, and the
results of \cite{neishtadt2008jump} can be applied. We leave for future works the relation of these results and the fascinating patterns seen in   Figure  \ref{fig:neartangent}. 
\begin{figure}[h]
\begin{centering}
\includegraphics[scale=0.5]{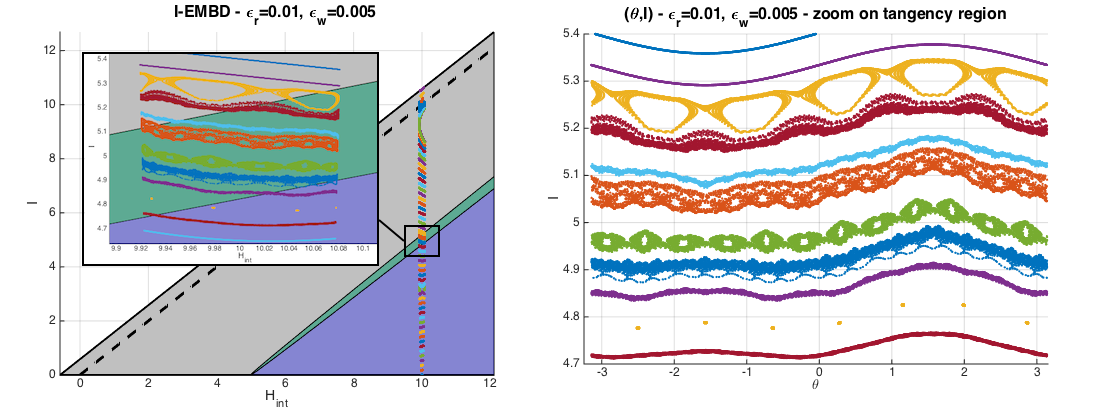}
\par\end{centering}

\protect\caption{\label{fig:neartangent}I-EMBD (right) and return map values (left) for near tangent initial conditions, indicated in the bottom left image in Figure \ref{fig:nearintresults} by a dashed rectangle (higher resolution of initial values is applied to the relevant region).}

\end{figure}

\appendix
\section*{Appendices}

\section{Boundedness of the perturbation terms on the energy surface}

\begin{lem}
\label{lem:pert-energy-surface}The perturbed energy surface corresponding
to a constant energy level set $h=H_{int}+\epsilon_{r}V_{r}$ is bounded.
\end{lem}
\begin{proof}
This is a result of the assumptions on the potential form and the implicit function theorem.
Note that since the system in consideration is mechanical,
i.e. $H=\Sigma_{i=1,2}\frac{p_{i}^{2}}{2}+V_{i}(q_{i})+\epsilon_{r}V_{r}(q_{1},q_{2})$,
it is enough to show that the Hill region - the allowed region of
motion in the configuration space, $(q_{1},q_{2})$ - is bounded.
Indeed, if the motion in $q_{1},q_{2}$ is restricted to a compact
Hill region, then from the assumptions on smoothness and boundedness the potential values $V_{1},V_{2},V_{r}$ are bounded and thus so are the momenta, hence the energy surface is bounded. 

Consider therefore the boundaries in $q$. These boundaries are the
potential level sets which define the Hill region, and are defined
by the equation $V_{int}(q_{1},q_{2})+\epsilon_{r}V_{r}(q_{1},q_{2})-h\equiv F(q_{1},q_{2};\epsilon_{r})=0$.
For $\epsilon_{r}=0$, by the S3B assumption, the solution of the equation $F(q_{1}^{0},q_{2}^{0};0)=0$
is a bounded region, i.e. there exists $R$
such that $||(q_{1}^{0},q_{2}^{0})||<R(h)$ for all $(q_{1}^{0},q_{2}^{0})$
on the energy surface $H_{int}=h$. Consider first energy levels which are bounded away from those containing fixed points of the integrable system (by the S3B assumption there are a finite number of such excluded energy intervals). Now, for $F(q_{1}^{\epsilon},q_{2}^{\epsilon};\epsilon)=0$,
since $||\mathbf{\nabla}V_{int}||>const.>0$
on such surfaces,  by the implicit function theorem  there exists
$\epsilon_{0}$ such that for all $\epsilon<\epsilon_{0}$, there exist solutions $(q_1^\epsilon,q_2^\epsilon)$ which are $\epsilon-$ close to $(q_1^0,q_2^0)$, and hence, for example, for sufficiently small $\epsilon_0$, $||(q_{1}^{\epsilon},q_{2}^{\epsilon})||<2R(h)$.
Hence, $q_{1},q_{2}$ are bounded and the Hill region is indeed compact.

Now consider the intervals of energy which contain points that are close to the extremal points
of the potential (the fixed points of the Hamiltonian system), i.e. where $||\mathbf{\nabla}V_{int}||=0$. The number of these points is finite and they
are contained in a bounded domain, from the assumed structure of the integrable Hamiltonian. Furthermore, as the Hill regions
$S_{h_{1}},S_{h_{2}}$ for different energy values $h_{1}<h_{2}$
are level sets of the potential function, these are nested regions
in the configuration space - $S_{h_{1}}\subset S_{h_{2}}$. Hence
the energy surfaces corresponding to singular energy level sets and their nearby energy surfaces are
bounded as well.\end{proof}

\section{Conditions I-V for the soft impact system}
Theorem 1 in \cite{RK2014smooth} establishes that finite segments of trajectories of the smooth impact Hamiltonian flow
\begin{equation}
H(q, p) =
\frac{p^{2}}{2} + U(q)+ V (q; \epsilon_b) ,
\end{equation}  with energy \(H<H_{max}\) limit to those of the hard impact system in some general bounded domain \(D \) in \(\mathbb{R}^d\) or \(\mathbb{T}^d\)  , in the \(C^{r}\) topology, provided these segments contain only regular reflections, and the potentials satisfy some general conditions; The potential \(V(q; \epsilon_b)\) is assumed to be a soft-billiard potential (satisfying conditions I-IV of   \cite{RK2014smooth,rapoport2007approximating}, which are also listed below). The smooth,  \(C^{r+1}\),  potential \(U\) is quite general - one only assumes that on the  domain boundary, which is assumed to be of finite length, \(U\)  is bounded from below by  \(\hat{U}>-\mathcal{E}\) (condition V in \cite{RK2014smooth}, see below), where \(\mathcal{E}\) denotes the limit of the billiard potential  energy at the wall (\(V (q; \epsilon_b)\) as \(q\rightarrow q^{w}\) and  \(\epsilon_{b}\rightarrow0\), so \(\mathcal{E}\) may be finite, similar to the example in Eq. (\ref{eq:finitebarrier}) or infinite as in Eq. (\ref{eq:infinitebarrier})). Finally, setting the maximal energy to \(H_{max}<\mathcal{E}+\hat{U}\)  insures that particles with \(H<H_{max}\) do not escape from the billiard domain. Here, we denote the form of the soft impact potential by $b\cdot{}V_{b}(\cdot,\epsilon_b)$ and adopt the convention that the barrier height $\mathcal{E}\geq{}b$ and can, again, be either finite of infinite.

To apply the above result to the current work we need to address only one issue - formally, for simplicity,  the conditions in \cite{RK2014smooth}
were stated for compact domains \(D \) with finite length boundary \(\partial D\), whereas here the domain \(D=\{(q_{1},q_2)|q_1\geqslant q_1^w=\epsilon_{w}Q^w(q_2),q_2\in\mathbb{R}\} \) is unbounded and has infinite length boundary \(\partial D=\{(q_{1}^{w}(q_{2}),q_2), q_{2}\in\mathbb{R}\}\). Noting that for finite energies \(H\), by the S3B assumption  on \(V_{int}\) and  \(V_{r}\), the Hill regions for all \(H\leq H_{max}\) are compact and are contained in the compact Hill region of \(H_{max}\)  (see appendix A), solves this formal problem; In particular, one can choose\begin{equation}
\hat{V}=\min_{q\in D_{Hill}(H_{max})}V_{int}(q)-1\leq\min_{q\in D_{Hill}(H_{max})}(V_{int}(q)+\epsilon_rV_r(q,\epsilon_r)) \label{eq:uhatdef}
\end{equation}  and the results of \cite{RK2014smooth} directly apply as long as the potential \(V_b(\cdot,\epsilon_{b})\) of (\ref{eq:hsoft}) is a billiard-like-potential on \(D \) (satisfies the conditions I-IV that are listed below on this  domain, with the billiard boundary set at \(q_1=q_1^w\)). In fact, it is sufficient to require that these conditions are satisfied on \(D\cap D_{Hill}(H_{max})\).  \( \)

The conditions I-V of  \cite{RK2014smooth}
 are listed below, almost verbatim: in some places notation is adjusted and simplified to the setting of the current paper, which is two-dimensional and has only one boundary component with no corners. Additionally, some remarks regarding the current setup are included. 

\textbf{Condition I.} For any fixed compact region $K\subset{}D$, the potential $V_b(q_1,q_2;\epsilon_w;\epsilon_b)$ diminishes along with all its derivatives as $\epsilon_b\rightarrow0$:
$$\lim_{\epsilon_b\rightarrow0}||V_b(q_1,q_2;\epsilon_w;\epsilon_b)\mid_{(q_1,q_2)\in{K}}||_{C^{r+1}}=0$$

We assume that the level sets of $V_b$ may be realized by some finite function near the boundary. Let $N$ denote the fixed (independent of $\epsilon_b$) neighborhood of the billiard boundary $\partial{}D$ (for example, here, \(N=\{q|q_1^w <q_{1}<0.1\}\). Assume that for all small $\epsilon_b\geq0$ there exists a \textit{pattern function}
$$Q(q_1,q_2;\epsilon_b):N\rightarrow\mathbb{R}^1,$$
which is $C^{r+1}$ with respect to $(q_1,q_2)$ in $N$ and depends continuously on $\epsilon_b$ (in the $C^{r+1}$ topology, so it has, along with all its derivatives, a proper limit as $\epsilon_b\rightarrow0$). 

Further assume that the following is fulfilled:

\textbf{Condition IIa.} The billiard boundary is a level surface of $Q(q_1,q_2;0)$:
$$Q(q_1,q_2;\epsilon_b=0)\mid_{(q_1,q_2)\in{}\partial{}D}\equiv{}\mathcal{Q}=const.$$

In the neighborhood $N$ of the barrier $\partial{}D$ (so $Q(q_1,q_2;\epsilon_b=0)$ is close to $\mathcal{Q}$), define a \textit{barrier function} $W(Q;\epsilon_b)$, which is $C^{r+1}$ smooth in $Q$, is continuous in $\epsilon_b$, and does not depend explicitly on $(q_1,q_2)$. Also assume that there exists $\epsilon_0$ such that conditions IIb-c are satisfied.

\textbf{Condition IIb.} For all $\epsilon_b\in{(0,\epsilon_0]}$ the potential level sets in $N$ are identical to the pattern function level sets, and thus
$$b\cdot{}V_b(q_1,q_;\epsilon_b)\mid_{(q_1,q_2)\in{N}}\equiv{}W(Q(q_1,q_2;\epsilon_b)-\mathcal{Q};\epsilon_b).$$

\textbf{Condition IIc.} For all $\epsilon_b\in{}(0,\epsilon_0]$, $\nabla{}V_b$ does not vanish in the finite neighborhood of the boundary surface $N$; thus

$$\nabla{}Q\mid_{(q_1,q_2)\in{}N}\neq0,$$
and for all $Q(q_1,q_2;\epsilon_b)\mid_{(q_1,q_2)\in{}N}$,
$$\frac{d}{dQ}W(Q-\mathcal{Q};\epsilon_b)\neq0.$$

Adopt the convention that $Q>\mathcal{Q}$ corresponds to the points near $\partial{}D$ inside the billiard.

\textbf{Condition III.} There exists a constant $\mathcal{E}>0$ ($\mathcal{E}$ may be infinite) such that as $\epsilon_b\rightarrow{}+0$ the barrier function increases from zero to $\mathcal{E}$ across the boundary $\partial{}D$:
$$\lim_{\epsilon_b\rightarrow{}+0}W(Q;\epsilon_b)=\begin{cases}
0, & Q>\mathcal{Q}\\
\mathcal{E}, & Q<\mathcal{Q}
\end{cases}$$

\textbf{Condition IV.} As $\epsilon_b\rightarrow+0$, for any fixed $W_1$ and $W_2$ such that $0<W_1<W_2<c$, the function $Q(W;\epsilon_b)$ tends to zero uniformly on the interval $[W_1,W_2]$ along with all of its $(r+1)$ derivatives.

For example, one can take here \(Q(q_1,q_2;\epsilon_w)=q_{1}-\epsilon _{w}Q^{w}(q_{2}),\) and  billiard-like potential  functions  of the form \(b\cdot{}V_{b}=W(Q(q;\epsilon_w),\epsilon_b)\), with \(W(Q,\epsilon_b)=b\cdot\exp(-Q/\epsilon_b)\) and \(W(Q,\epsilon_b)=-b\cdot\epsilon_b/Q\) corresponding to the billiard potential (\ref{eq:infinitebarrier}),(\ref{eq:finitebarrier}) respectively, and $\mathcal{E}\geq{}b$. According to Theorem 1 of  \cite{RK2014smooth}, we can choose \(\epsilon_{w}\) to depend on \(\epsilon_b\) or take them independent.

The last condition is concerned with the addition of the smooth component of the potential $U(q)$ assuring that together with the billiard-like potential, particles that are initially in $D$ cannot escape. Defining  
 \(\hat{U}=\min_{q\in\partial D\ }U(q)\)
one assumes that:

\textbf{Condition V.} $U(q)$ is a $C^{r+1}$ smooth potential bounded in the $C^{r+1}$ topology on an open set $\mathcal{D}$, where $\bar{D}\in{}\mathcal{D}$. The minimum of $U$ on the boundary $\partial{}D$ satisfies $\hat{U}>-\mathcal{E}$.

Using in the above definition  $U(q)=V_{int}(q)+\epsilon_rV_r(q)$ 
it suffices in our setting to require $\hat{U}>-b$. In particular, Theorem 1 of   \cite{RK2014smooth}  then applies to trajectory segments with bounded energies which satisfy \(H^{}<H_{max}(b)=b+\hat{V}=b+\min_{q\in D_{Hill}(H_{max})}V_{int}-1 \) where \(b\) is the minimal billiard potential barrier height. 
\bibliography{KAM-type-bib}
\end{document}